\documentclass[a4paper,11pt,fleqn]{article}

\usepackage{amsmath,amssymb,amsthm}
\usepackage[mathscr]{eucal}
\usepackage{graphicx}

\allowdisplaybreaks
\makeatletter
\long\def\@makecaption#1#2{%
  \vskip\abovecaptionskip\footnotesize
  \sbox\@tempboxa{#1. #2}%
  \ifdim \wd\@tempboxa >\hsize
    #1. #2\par
  \else
    \global \@minipagefalse
    \hb@xt@\hsize{\hfil\box\@tempboxa\hfil}%
  \fi
  \vskip\belowcaptionskip}
\makeatother

\newcommand{\todo}[1][\null]{\ensuremath{\clubsuit}}

\newcommand{\noprint}[1]{}
\newcommand{\checked}[1][\null]{\ensuremath{\boldsymbol{\surd}}}

\newtheorem{theorem}{Theorem}
\newtheorem{lemma}[theorem]{Lemma}
\newtheorem{corollary}[theorem]{Corollary}
\newtheorem{proposition}[theorem]{Proposition}
\newtheorem*{problem*}{Problem}
{\theoremstyle{definition}
\newtheorem*{notation*}{Notation}
\newtheorem{definition}[theorem]{Definition}
\newtheorem{example}[theorem]{Example}
\newtheorem{remark}[theorem]{Remark}
\newtheorem*{remark*}{Remark}
}

\newcommand{\p}{\partial}

\newcommand{\ord}{\mathop{\rm ord}\nolimits}
\newcommand{\cancel}{\lefteqn{\smash{\mbox{\large$\diagup$}}}}  

\begin{document}

\par\noindent {\LARGE\bf
Inverse problem on conservation laws
\par}

\vspace{4mm}\par\noindent{\large
Roman O.\ Popovych$^\dag$ and Alexander Bihlo$^\ddag$
}

\vspace{4mm}\par\noindent {\it
$^{\dag}$Wolfgang Pauli Institute, Oskar-Morgenstern-Platz 1, 1090 Vienna, Austria\\
$\phantom{^\ddag}$Institute of Mathematics of NAS of Ukraine, 3 Tereshchenkivska Str., 01601 Kyiv, Ukraine
}

\vspace{2mm}\par\noindent {\it
$^\ddag$Department of Mathematics and Statistics, Memorial University of Newfoundland,\\
$\phantom{^\ddag}$St.\ John's (NL) A1C 5S7, Canada
}

\vspace{4mm}\par\noindent
\textup{E-mail:} rop@imath.kiev.ua, abihlo@mun.ca

\vspace{6mm}\par\noindent\hspace*{5mm}\parbox{150mm}{\small
The explicit formulation of the general inverse problem on conservation laws is presented for the first time. In this problem one aims to derive the general form of systems of differential equations that admit a prescribed set of conservation laws. The particular cases of the inverse problem on first integrals of ordinary differential equations and on conservation laws for evolution equations are studied. We also solve the inverse problem on conservation laws for differential equations admitting an infinite dimensional space of zeroth-order conservation-law characteristics. This particular case is further studied in the context of conservative first-order parameterization schemes for the two-dimensional incompressible Euler equations. We exhaustively classify conservative first-order parameterization schemes for the eddy-vorticity flux that lead to a class of closed, averaged Euler equations possessing generalized circulation, generalized momentum and energy conservation.
\par}\vspace{4mm}

\noprint{

Title    Inverse problem on conservation laws
Authors  Roman O. Popovych and Alexander Bihlo
MSC     35A30 (Primary) 37K05, 76M60, 86A10 (Secondary)          37K05 76M60 86A10

}

\section{Introduction}

Conservation laws play a distinguished role in mathematical physics. They have multiple practical applications in several areas related to differential equations, including integrability theory, asymptotic integrability and the construction of geometric numerical integration schemes.

There is a vast body of literature devoted to the so-called direct problem on conservation laws. Here one is given a system of differential equations and aims to find its space of conservation laws, or at least a subspace of this space singled out by additional constraints, such as a prescribed upper bound for the order of conservation laws to be considered.
Standard tools for the solution of the direct problem on conservation laws include Noether's theorem, different variations of the direct method and techniques based on co-symmetries, see~\cite{anco02Ay,anco02By,blum10Ay,boch99Ay,olve93Ay,popo2008a,vino84Ay,wolf02Ay} and references therein.
For a class of (systems of) differential equations,
i.e., for a family of systems of differential equations that is parameterized by a tuple of arbitrary functions or constant parameters, 
collectively referred to as arbitrary elements \cite{ovsi82Ay,popo10Ay}, 
one should tackle the direct problem on conservation laws as classification problem
since then the space of conservation laws in general depends on the arbitrary elements parameterizing systems of the class.

The direct classification problem on conservation laws is in many aspects similar to the direct (symmetry) group classification of differential equations.
Directly classifying Lie symmetries in a given class of differential equations,
one aims to find, up to equivalence, those systems that admit more symmetries than the most general system from the class \cite{bihl11Dy,ovsi82Ay,popo10Cy,popo10Ay}.
The associated inverse problem on group classification is well investigated too
\cite{Fushchych&Nikitin1994en,Fushchych&Shtelen&Serov1993en,olve93Ay,ovsi82Ay,popo10Cy}.
Here one is given a Lie group and finds those systems of differential equations admitting the selected group as a symmetry group.
This problem deserves attention due to the important role that symmetries play in the mathematical sciences; virtually all central models of modern physics are invariant under wide symmetry groups and hence the classification of systems of differential equations invariant under prescribed Lie groups is a significant direction of the study in the field of group analysis.
The inverse problem of group classification is solved by systematically computing differential invariants with infinitesimal methods~\cite{olve93Ay,ovsi82Ay} 
or with equivariant moving frames~\cite{cheh08Ay,fels98Ay,fels99Ay,olve09Ay}.
Recently the study of inverse symmetry problems was extended to generalized conditional symmetries of evolution systems in (1+1)-dimensions~\cite{serg02a,serg06a}.
Since conservation laws take up a distinguished place in physical theories as well, the inverse problem on conservation laws is also relevant from the physical point of view.

The inverse problem on conservation laws has received less attention so far although it has several important fields of applications as well. 
The intuitive formulation of the inverse problem is the following:

\begin{problem*}
Derive the general form of systems of differential equations with a prescribed set of conservation laws.
\end{problem*}

The inverse problem was considered in~\cite{gali79a} for the particular case of single evolution equations,
where both equations and densities of their conservation laws do not explicitly involve the corresponding independent variables.
As far as we know, this was the first work on the inverse problem on conservation laws in general,
and this is the only paper on the subject in the literature, which has no essential citations.

The inverse problem on conservation laws arises naturally in the construction of physical parameterization schemes (or closure models), e.g., in geophysical fluid mechanics. Parameterization refers to the procedure of including unresolved processes into numerical models of the Earth system. 
Constructing closure models for such processes is one of the main directions of present research in the geosciences~\cite{sten07Ay}. 
The problem is timely as even a continuous increase of resolution in new numerical models cannot resolve all dynamically active scales that govern the time evolution of the Earth system. The construction of sensible closure models for these unresolved scales is hence of major importance to continue improving numerical simulations for weather and climate processes.
In practice this is done by averaging the governing equations of hydro-thermodynamics in space and time, with the length of the averaging intervals being determined by the resolution of the resulting numerical model of these equations. Averaging nonlinear equations introduces unresolved terms that have to be expressed using the resolved values. This is the parameterization step, which by its nature is always an imperfect approximation. Physically, these unresolved terms correspond to processes that cannot be resolved at the particular resolution of the resulting numerical model for the governing equations of hydro-thermodynamics. Mathematically, parameterization is done by replacing an unclosed system of differential equations with a class of closed systems of differential equations. 
The inverse problem on conservation laws arises as the problem of choosing the arbitrary elements of the class in such a manner that the corresponding closed systems inherit, in a certain sense, some of the conservation laws of the original system. Note that this is not an artificial problem. Parameterization schemes in practice are routinely constructed so as to preserve some conservation laws of the initial system of differential equations. One prominent example is the construction of parameterization schemes for the two-dimensional Euler equations. Physical considerations show that the parameterized two-dimensional Euler equations should be energy preserving. Thus, the inverse problem on conservation laws for the two-dimensional Euler equations could be formulated as the problem of finding the class of closed systems arising from parameterizing the averaged two-dimensional Euler equations, such that all members of the class admit at least conservation of energy. We will address this problem in detail in Section~\ref{sec:ConservativeParameterizationVorticityEquation}. For a more detailed introduction to the problem of parameterization and its relation to classes of differential equation, we include a short, self-contained description in Section~\ref{sec:ConservativeParameterizationSchemes}.

The inverse problem on conservation laws is also relevant in the context of geometry-preserv\-ing discretization, i.e., if numerical discretizations are sought that should preserve conservation laws of the associated differential equations. This problem is of high practical relevance in fields that require long time integrations as the preservation of conservation laws is usually mandatory for such applications~\cite{hair06Ay}.

Solving the inverse problem on conservation laws for a class of differential equations may help in the solution of the direct problem for this class.
Methods related to the inverse problem on conservation laws were recently used to classify all conservation laws for the class of (1+1)-dimensional even-order evolution equations~\cite{popo2019a}.

The precise formulation of the inverse problem on conservation laws requires considerable work. In particular, conservation laws are defined as equivalence classes of conserved currents and correspond to equivalence classes of their characteristics. Hence one needs to determine what constitutes the appropriate data for this problem. Thus, the inverse problem on conservation laws is more involved than the related inverse problem of group classification.

The interpretation of the inverse problem on conservation laws may be extended considering not the precise form of individual conservation laws but their properties.

\begin{problem*}
Study properties of systems of differential equations that are implied by prescribed properties of conservation laws of these systems.
\end{problem*}

A number of known results on conservation laws of differential equations can be interpreted  
within the framework of the extended inverse problem on conservation laws, 
and some ``if and only if'' statements simultaneously solve a direct problem and its inversion. 
An example of intertwining the direct and inverse problem on conservation laws in the extended interpretation is given by 
the generalized version~\cite{popo2019b} of Noether's second theorem \cite[Theorem 5.66]{olve93Ay}. 
It states that a system of differential equations possesses 
conservation laws parameterized by an arbitrary smooth functions of all independent variables 
if (the direct problem) and only if (the inverse problem) this system is abnormal.%
\footnote{%
A system of differential equations is called \emph{abnormal} if it has an identically vanishing differential consequence
such that the associated tuple of differential operators acting on system's equations does not vanish on the manifold
defined by the system and its differential consequences in the corresponding infinite-order jet space.
Such systems are also called globally under-determined \cite[p.~170--171]{olve93Ay}.%
}
Abnormal systems are also characterized by the presence of trivial conserved currents associated with nontrivial characteristics.

The problem on reversing Noether's first theorem, 
which was posed by~\cite{take1977a} and later tackled in \cite{ande1994b,ande1995c}, 
can be considered as a joint inverse problem on conservation laws and Lie symmetries. 
Roughly speaking, the statement of Takens' problem is the following. 
Given a system of differential equations~$\mathcal L$ admitting a Lie algebra~$\mathfrak g$ of Lie symmetries 
simultaneously being conservation-law characteristic of~$\mathcal L$, 
which properties the system~$\mathcal L$ and the algebra~$\mathfrak g$ should have to guarantee 
that $\mathcal L$ is the system of Euler--Lagrange equation for a Lagrangian?
In~\cite{ande1994b,ande1995c,take1977a}, the generalized Noether's first theorem 
that includes both Noether's first theorem and its converse was proved 
for scalar second-order partial differential equations, for linear systems or, more generally, 
for systems of equations that are polynomial in the dependent variables and their derivatives
and whose total degree is not greater than the pointwise dimension
of the projection of the algebra~$\mathfrak g$ to the space of independent variables.
Results related to the converse to Noether's theorem were earlier obtained in classical field theories 
\cite{horn1974a,love1971a,love1974a}, 
where the corresponding Lie symmetry group is the infinite-dimensional group 
of all coordinate transformations on the underlying base manifold.
See also \cite{ande2012a,mann2008} and references therein for the further development of the subject.

The classification of conservation laws for a class of differential equations can often be reformulated 
as the solution of a series of related inverse and direct problems on conservation laws. 
Thus, in the course of classification of conservation laws 
of non-Goursat second-order parabolic partial differential equations 
for one function of two independent variables in~\cite{brya1995a}, 
such equations admitting at least one-, two- and three-dimensional spaces of conservation laws were described, 
and such equations with at least four-dimensional spaces of conservation laws were proved to be linearizable. 
In particular, only Monge--Amp\`ere parabolic equations among the above ones were shown to possess conservation laws. 
A similar study for second-order parabolic partial differential equation for one function of three independent variables
was carried out in \cite{clel1997a,clel1997b}. 
Conservation laws of (1+1)-dimensional second-order evolution equations that are invariant with respect to time translations
were classified up to contact transformations in~\cite{brya1995a} as well. 
This classification was extended in \cite{popo2008d} to the entire class of (1+1)-dimensional second-order evolution equations.  
Normal forms of such equations with one- and two-dimensional spaces of conservation laws were derived up to contact transformations therein, 
and equations with conservation-law spaces of greater dimensions were shown to be linearizable. 
Normal forms of a (1+1)-dimensional evolution equation of arbitrary order 
that possesses a conservation law of order not greater than one or a pair of linearly independent zeroth-order conservation laws 
were obtained up to contact equivalence in~\cite{svin1985a} and in~\cite{popo2010a}, respectively; 
see also~\cite{folt2002a} for particular cases of these forms.
An inverse problem on conservation laws naturally arises in the course of looking for integrable systems of evolution equations 
using the concept of formal symmetry~\cite{mikh1991a}.

The further organization of this paper is as follows: 
In Section~\ref{sec:ConservationLaws} we present some essential facts on conservation laws of differential equations.
Section~\ref{sec:InverseProblemOnCLs} is devoted to the correct formulation of the inverse problem on conservation laws.
We show that conservation-law characteristics are in general the more appropriate data for the inverse problem on conservation laws than the entire conserved currents.
We also state the well-posedness of the inverse problem on conservation laws for the wide class of systems of differential equations in the extended Kovalevskaya form,
where the prescribed data are reduced conservation-law characteristics or reduced conservation-law densities.
A particular case, which is the inverse problem on first integrals for single ordinary differential equations, is comprehensively studied in Section~\ref{sec:InverseProblemCLODE}.
In Section~\ref{sec:InverseProblemCLsEvolutionEquations} we review the inverse problem on conservation laws for single (1+1)-dimensional evolution equations. 
%
In Section~\ref{sec:ConservativeParameterizationWithDirectMethods}, for the case of single partial differential equations we solve the inverse problem on conservation laws with characteristics that are arbitrary smooth functions of a single variable (or several variables) or have a related simple structure. Conservation laws of this kind are typical, in particular, for systems of differential equations modeling flows of an incompressible fluid, such as the incompressible Euler equations, the incompressible Navier--Stokes equations and the vorticity equation.
Notions involved in the framework of conservative parameterization and the interpretation of conservative parameterization as an inverse problem on conservation laws are discussed in Section~\ref{sec:ConservativeParameterizationSchemes}. In Section~\ref{sec:ConservativeParameterizationVorticityEquation} we use the results proved in Section~\ref{sec:ConservativeParameterizationWithDirectMethods} to derive the general form of parameterization schemes for the eddy-vorticity flux in the Reynolds averaged two-dimensional incompressible Euler equations that preserve generalized circulation, generalized momenta in $x$- and $y$-directions as well as energy. The research perspective on further required investigations within the inverse problem on conservation laws and the conclusions are found in Section~\ref{sec:ConclusionCP}.

\section{Conservation laws of differential equations}\label{sec:ConservationLaws}

It is appropriate to collect here a few results related to conservation laws of differential equations and their relation to the parameterization problem. A more extensive account of this material can be found e.g.\ in~\cite{blum10Ay,olve93Ay,popo2008a}.

Here and in the following we denote by $\mathcal L$ a system of differential equations. The system $\mathcal L$ consists of $l$ equations of the form $L^\mu(x,u_{(r)})=0$, $\mu=1,\dots,l$, where $x=(x_1,\dots,x_n)$ are the $n$~independent variables, $u=(u^1,\dots,u^m)$ are the $m$~unknown functions (the dependent variables) and the symbol $u_{(r)}$ denotes all derivatives of the functions $u$ with respect to $x$ of order not greater than $r$. By definition $u$ is included in $u_{(r)}$ as derivatives of order zero.
Within the local approach, which is employed in the present paper,
differential equations can be interpreted as algebraic%
\footnote{%
Here the adjective ``algebraic'' is only used in the sense ``non-differential'';
it does not mean that equations are polynomial; cf.\ \cite[Section~2.1]{olve93Ay}.
}
equations in the jet space $\mathrm J^{\infty}(x|u)$,
where both the independent variables~$x$ and the derivatives of~$u$ with respect to~$x$ are assumed as usual variables.
A smooth function~$f$ depending on~$x$ and a finite number of derivatives of~$u$
(i.e., a smooth function on an open set of $\mathrm J^{\infty}(x|u)$ with finite number of arguments and with values in the underlying field)
is called a \emph{differential function} of~$u$, which is denoted by $f=f[u]$.
The order~$\ord f$ of a differential function~$f$ is the highest order of derivatives involved in~$f$,
and, if~$f$ does not depend on derivatives of~$u$, $\ord f=-\infty$.

\begin{definition}\label{def:DefinitionConservedVectorCP}
A \emph{conserved current} of the system $\mathcal L$ is an $n$-tuple of differential functions $F=(F^1[u],\dots,F^n[u])$
the total divergence of which vanishes on the solutions of~$\mathcal L$,
\begin{equation}\label{eq:ConsLawDefinition}
(\mathop{\rm Div}F)\big|_{\mathcal L}=0.
\end{equation}
\end{definition}

\begin{notation*}
In Definition~\ref{def:DefinitionConservedVectorCP} and in what follows,
the total divergence operator is defined by $\mathop{\rm Div}F=\mathrm D_iF^i$,
and $\mathrm D_i=\mathrm D_{x_i}$ denotes the operator of total derivative with respect to the variable~$x_i$.
In other words, $\mathrm D_i=\partial_i+u^a_{\alpha+\delta_i}\partial_{u^a_\alpha}$,
where $\alpha=(\alpha_1,\dots,\alpha_n)$ is an arbitrary multi-index, $\alpha_i\in\mathbb N_0=\mathbb N\cup\{0\}$,
the index~$i$ runs from 1 to~$n$, the index~$a$ runs from 1 to~$m$,
the variable~$u^a_{\alpha}$ of the jet space $\mathrm J^{\infty}(x|u)$ is identified with the derivative
$\partial^{|\alpha|}u^a/\partial x_1^{\alpha_1}\cdots \partial x_n^{\alpha_n}$,
$|\alpha|=\alpha_1+\cdots+\alpha_n$, and $\delta_i$ is the multi-index with zeros everywhere except on the $i$th entry, which equals $1$.
$\partial_i:=\partial/\partial x_i$ and $\partial_{u^a_\alpha}:=\partial/\partial u^a_\alpha$.
The summation convention over repeated indices is used.
With $(\dots)\big|_{\mathcal L}=0$ we mean that the corresponding expression only vanishes for solutions of the system~$\mathcal L$.
\end{notation*}

The validity of~\eqref{eq:ConsLawDefinition} on the solution set of~$\mathcal L$ is significant for relating the conserved current $F$ to~$\mathcal L$.
A conserved current $F$ is \emph{trivial} if it is represented as $F=\hat F+\check F$, where $\hat F$ and $\check F$ are $n$-tuples of differential functions such that the components of $\hat F$ vanish on the solutions of $\mathcal L$ and $\check F$ is a null divergence.
By null divergence it is meant that $\mathop{\rm Div}\check F=0$ holds unrestricted of the system $\mathcal L$.

Two conserved currents $F$ and $F'$ are called \emph{equivalent} if their difference $F-F'$ is a trivial conserved current.
It is obvious that for any system~$\mathcal L$ its set of conserved currents, denoted by $\mathrm{CC}{(\mathcal L)}$, is a linear space. Likewise, the subset of trivial conserved currents, denoted by $\mathrm{CC}_0{(\mathcal L)}$, is a linear subspace of $\mathrm{CC}{(\mathcal L)}$. The set of equivalence classes of $\mathrm{CC}{(\mathcal L)}$ with respect to the above equivalence relation on conserved currents is the quotient space $\mathrm{CC}{(\mathcal L)}/\mathrm{CC}_0{(\mathcal L)}$, which is denoted by~$\mathrm{CL}(\mathcal L)$.

\begin{definition}
 The linear space $\mathrm{CL}{(\mathcal L)}$ is called the \emph{space of (local) conservation laws} of the system~$\mathcal L$.
 Its elements are called \emph{(local) conservation laws} of the system~$\mathcal L$.
\end{definition}

In other words, equivalent conserved currents correspond to the same conservation law.

If the system $\mathcal L$ is totally nondegenerate~\cite{olve03Ay} or weakly totally nondegenerate~\cite{kunz08a,popo2019b},
then it is possible to use the Hadamard lemma and `integration by parts' to represent the definition of conserved current~\eqref{eq:ConsLawDefinition} in the form
\begin{equation}\label{eq:CharacteristicDefinition}
 \mathop{\rm Div} F=\lambda^1 L^1+\dots+\lambda^l L^l,
\end{equation}
where the initial conserved current~$F$ should be replaced by one differing from~$F$ in a trivial conserved current,
$F+\hat F\to F$, where the components of $\hat F$ vanish on the solutions of $\mathcal L$.

\begin{definition}
The $l$-tuple of differential functions $\lambda=(\lambda^1,\dots,\lambda^l)$ is called the \emph{characteristic} and Eq.~\eqref{eq:CharacteristicDefinition} is the \emph{characteristic form} of the conservation law corresponding to the conserved current $F$.
\end{definition}

 The \emph{Euler operator} $\mathsf E=(\mathsf E^1,\dots,\mathsf E^m)$ is the $m$-tuple of differential operators defined by
 \[
  \mathsf E^a=(-\mathrm D)^\alpha\partial_{u^a_\alpha},\quad a=1,\dots,m,
\quad\mbox{where}\quad (-\mathrm D)^\alpha= (-\mathrm D_1)^{\alpha_1}\cdots(-\mathrm D_n)^{\alpha_n}.
 \]
It is well known~\cite[Theorem 4.7]{olve93Ay} that a differential function $f$ is (locally) a total divergence, meaning that $f=\mathop{\rm Div}F$ for some $n$-tuple of differential functions~$F$, if and only if it is annihilated by the Euler operator, $\mathsf E^af=0$. In other words, $\mathop{\rm im}\mathop{\rm Div}=\mathrm{ker}\,\mathsf E$ (locally). Using this property of the Euler operator and applying it to the characteristic form of conservation laws~\eqref{eq:CharacteristicDefinition}, one obtains
\begin{equation}\label{eq:DeterminingEquationsCharacteristics}
   \mathsf E^a(\lambda^1 L^1+\dots+\lambda^l L^l)=0,
\end{equation}
which is both a necessary and sufficient condition for the tuple $\lambda$ to be a conservation-law characteristic of the system~$\mathcal L$. The characteristic approach to conservation laws is particularly suitable for the automatic computation of low-order conservation laws for systems of differential equations in the extended Kovalevskaya form using computer algebra systems, see e.g.~\cite{chev07Ay,wolf02Ay}.

 The notion of triviality extends to conservation-law characteristics as well. A characteristic~$\lambda$ is called trivial if it vanishes for all solutions of $\mathcal L$. The existence of trivial characteristics makes it necessary to introduce equivalent characteristics. If the difference $\tilde\lambda-\lambda$ of characteristics $\lambda$ and $\tilde\lambda$ is a trivial characteristic, then the characteristics~$\lambda$ and~$\tilde\lambda$ are called \emph{equivalent}. Similar as for conserved currents the set of characteristics, denoted by $\mathrm{Ch}(\mathcal L)$, is a linear space with the subset $\mathrm{Ch}_0(\mathcal L)$ of trivial characteristics being a linear subspace thereof. For a normal totally nondegenerate system~$\mathcal L$, the characteristic form of conservation laws~\eqref{eq:CharacteristicDefinition} then induces a one-to-one correspondence between the factor spaces $\mathrm{CC}{(\mathcal L)}/\mathrm{CC}_0{(\mathcal L)}$ and $\mathrm{Ch}(\mathcal L)/\mathrm{Ch}_0(\mathcal L)$. This correspondence forms the basis of both Noether's theorem and the direct method for construction of conservation laws as found in~\cite{anco02Ay,anco02By,blum10Ay,boch99Ay,olve93Ay,vino84Ay}.

\section{Statement of inverse problem on conservation laws}\label{sec:InverseProblemOnCLs}

Having introduced some of the necessary background on conservation laws we now proceed with the proper statement of the inverse problem on conservation laws. This statement appears considerably more difficult than the statement of the analogous inverse problem of group classification, which has been the subject of extensive investigations, see~\cite{bihl11Fy,olve93Ay,ovsi82Ay,popo10Cy} for discussions and some physical applications.

The first step for the formulation of the inverse problem on conservation laws is to determine which data of conservation laws should be used.
Here we justify why the appropriate choice is to invoke densities and characteristics but not entire conserved currents
and, moreover, why the starting point of the consideration rests, in most cases, on a generalization of the Kovalevskaya form for systems of differential equations.

\medskip\par\noindent
{\bf Conserved currents.}
As conservation laws are equivalence classes of conserved currents, it seems at first sight that a fixed conserved current might be chosen as the appropriate datum for the corresponding conservation law. There are, however, at least two counter-arguments against doing this.

The equivalence of conserved currents is of complex structure involving two kinds of trivial conserved currents.
In contrast to characteristics, it is not so obvious to determine whether a conserved current is trivial or not even for a fixed system of differential equations.
In fact this determination reduces, at least implicitly, to the consideration of related characteristics.
For classes of systems of differential equations, the situation is even more complicated.
Conserved currents may be trivial for some systems from the class and nontrivial for other systems from the same class.
In other words, it is then difficult to test whether such data are really independent to each other and thus whether the problem with such data is well posed.

The use of entire conserved currents in the discussed framework is also prevented by the fact
that not all $n$-tuples of differential functions are suitable candidates for conserved currents 
if the number of independent variables~$n$ is greater than~one.%
\footnote{%
For example, for a system of differential equations in the extended Kovalevskaya form~\eqref{Eq:ExtendedCauchyKovalevskayaForm} with $n>1$,
any \mbox{$n$-tuple} of differential functions with vanishing $n$th component is
either a trivial conserved current or not a conserved current at all.
}
In the case of multiple conservation laws for single differential equations
the conditions arising for equations from different candidates for conserved currents may be inconsistent.
Moreover, fixing conserved currents considerably restricts the form of the corresponding system, especially in the case of single equations, see Remark~\ref{rem:OnConservedCurrentsForEvolutionEquations}.
Even minor variations within the form of systems from a class lead to varying the form of conserved currents.

The above counter-arguments lead to the idea that a part of the components can be used in the statement of the inverse problem on conservation laws
instead of entire conserved currents.

\medskip\par\noindent
{\bf Characteristics.}
Another class of objects related to conservation laws and describing them consists of conservation-law characteristics.
There are several \emph{arguments} for working with characteristics  rather than entire conserved currents.

Thus, conservation-law characteristics are equivalent if and only if they coincide on solutions of the considered system of differential equations, which is much simpler than the equivalence relation for conserved currents. This is why the equivalence of characteristics can be easily verified, even for classes of differential equations.

Another argument is that characteristics are, roughly speaking, more stable under varying systems within a class of differential equations. While conserved currents change under any modification of the corresponding system of differential equations, there is a chance that some characteristics are not modified under such changes, thus being the more appropriate object for the study of the inverse problem.

Finally, for variational systems characteristics of variational symmetries are characteristics of conservation laws. Therefore, the inverse problem on conservation laws has an immediate connection with the inverse group classification problem for such systems, which is directly exploitable if the problem is formulated in terms of characteristics.

While there are also some \emph{obstacles} for working with characteristics, they are not too principal and can be controlled.

In particular, characteristics can be defined only for systems that are weakly totally nondegenerate~\cite{olve93Ay,popo2019b}, but this property is quite natural and the class of weakly totally nondegenerate systems is quite wide.

For characteristics to be perfect initial data for the inverse problem on conservation laws, the equivalences of conserved currents and characteristics should be compatible, i.e., for systems of differential equations under consideration we should have a one-to-one correspondence between conservation laws and equivalence classes of characteristics. However, this correspondence was only proven for normal totally nondegenerate systems~\cite[Theorem 4.26]{olve93Ay}, see also~\cite{mart79a} for the first formulation of this result. The above one-to-one correspondence does not exist, e.g.\ for abnormal systems. In view of the generalized second Noether theorem, each abnormal system admits nontrivial characteristics that correspond to trivial conserved currents~\cite{popo2019b}, see also \mbox{\cite[p.~345]{olve93Ay}} for the case of Euler--Lagrange equations. At the same time, such lack of one-to-one correspondence is again not principal for the purpose of the inverse problem on conservation laws. What is essential is that trivial characteristics should be associated with trivial conserved currents, or, in other words, equivalent characteristics correspond to the same conservation law. The existence of systems without this property is still an open problem. Such systems should be quite artificial and they do not arise in practical applications. Hence, the correspondence between characteristics and conservation laws makes no serious complication for the inverse problem on conservation laws.

The main downside of working with conservation-law characteristics rather than with conserved currents is that characteristics are not invariant under \emph{system equivalence}. Recall that two systems of differential equations are equivalent if they are defined on the same space of independent and dependent variables and can be obtained from each other using the following operations:
 \begin{itemize}\itemsep=0ex
  \item Recombining equations with coefficients that are differential functions and constitute a nondegenerate matrix (giving rise to \emph{linearly equivalent systems})
  \item Supplementing a system with its differential consequences or excluding equations that are differential consequences of other equations.
 \end{itemize}
It is immediately obvious that equivalent systems have the same solution set and hence the same set of conserved currents and the same sets of various kinds of symmetries, but this is not the case for both co-symmetries and conservation-law characteristics.
Fortunately, the system equivalence is not too relevant for the inverse problem on conservation laws. If a proper class of systems of differential equations without gauge equivalence is chosen, each set of equivalent systems has exactly one representative in the class~\cite{lisl92Ay,popo10Ay}. Even if the gauge equivalence is nontrivial, it can be made inessential since in any case the number of system equations as well as the number of independent and dependent variables are fixed within the class. This is why characteristics are in any case the appropriate initial data for the inverse problem on conservation~laws.

One more obstacle for working with conservation-law characteristics is created by the process of confining characteristics to solutions of the considered systems. In contrast to symmetries, co-symmetries and conserved currents, if a tuple of differential functions coincides with a conservation-law characteristic for all solutions of a system~$\mathcal L$, it does not mean that this tuple itself is a conservation-law characteristic of~$\mathcal L$. In particular, the exclusion of principal derivatives of~$\mathcal L$ from a conservation-law characteristic of~$\mathcal L$ may give a co-symmetry that is not a conservation-law characteristic of~$\mathcal L$.

\medskip\par\noindent
{\bf Form for systems.}
The starting point for posing a particular inverse problem on conservation laws is to specify the number of independent and dependent variables and the general form of the system including the number of equations to be considered. In order to overcome the aforementioned obstacles with conserved currents and characteristics, a good (and, in the most general settings, unique) choice for the general form is a Kovalevskaya form. At the same time, requiring systems of differential equations to be of the Kovalevskaya form \cite[pp.~162--163]{olve93Ay} is too restrictive for the inverse problem on conservation laws. The following more general form first introduced in~\cite{chev17a} can be used in this framework.

\begin{definition}
 A system of partial differential equations $\mathcal L$ is of \emph{extended Kovalevskaya form} if its equations can be written as
 \begin{equation}\label{Eq:ExtendedCauchyKovalevskayaForm}
   u^a_{r_a\delta_n}:=\frac{\p^{r_a}u^a}{\p x^{r_a}_n}=H^j(x,\widetilde{u_{(r)}}),\quad a=1,\dots,m.
 \end{equation}
 where $0\leqslant r_a\leqslant r$ and $\widetilde{u_{(r)}}$ denotes all derivatives of the functions $u$ with respect to $x$ up to order~$r$, where each $u^b$ is differentiated with respect to~$x_n$ at most $r_b-1$ times, $b=1,\dots,m$.
\end{definition}

In terms of Riquier's compatibility theory, all derivatives appearing on the left-hand side in~\eqref{Eq:ExtendedCauchyKovalevskayaForm} and their differential consequences are called {\it principal derivatives} for~$\mathcal L$. The other derivatives are called {\it parametric derivatives} of~$\mathcal L$.

The extension of the Kovalevskaya form is necessary for a few reasons.%
\footnote{%
The extended Kovalevskaya form is much less restrictive than the usual Kovalevskaya form, since two conditions are weakened. Firstly, zeroth-order derivatives may appear in the left-hand side of equations. Secondly, there are no restrictions on the order of derivatives with respect to $x_1,\dots,x_{n-1}$ depending on $r_a$'s.
Systems of the form~\eqref{Eq:ExtendedCauchyKovalevskayaForm} with positive $r_a$'s are called \emph{normal systems} in~\cite{mart79a}
and \emph{Cauchy--Kowalevsky systems in a weak sense} (resp., \emph{pseudo CK systems} in short) in~\cite{tsuj82a}.
}
Physically relevant equations such as the (1+2)-dimensional linear heat equation $u_t=u_{xx}+u_{yy}$ are not conveniently represented in the standard Kovalevskaya form as the representation $u_{xx}=u_t-u_{yy}$ is not natural from the physical point of view since the density should then be associated with~$x$ instead of~$t$. Also, the extended Kovalevskaya form is a natural representation for potential systems corresponding to conservation laws constructed with conservation laws of (1+1)-dimensional differential equations.

Below we collect all facts on that the extended Kovalevskaya form allows one to select the data for the inverse problem on conservation laws in a proper way.

Using a part of conserved current components as initial data for the inverse problem on conservation laws instead of entire conserved currents naturally singles out the independent variables associated with these components and, in turn, the derivatives with respect to these variables in the corresponding system. This aligns well with the extended Kovalevskaya form~\eqref{Eq:ExtendedCauchyKovalevskayaForm}. The distinguished independent variable is~$x_n$.
Each conserved current of~\eqref{Eq:ExtendedCauchyKovalevskayaForm} is naturally split into the density (which is the component associated with the variable $x_n$) and the flux (consisting of components associated with the variables $x_1$, \dots, $x_{n-1}$), and the density completely defines the corresponding conservation law. This means that only densities are the proper parts of conserved currents as initial data. If one tries to use more components of a conserved current simultaneously then the problem may become inconsistent.

The notion of characteristics can be introduced only for systems of special kind, i.e., for weakly totally nondegenerate systems~\cite{popo2019b}, and totally locally solvable systems of extended Kovalevskaya form are definitely of this kind. The extended Kovalevskaya form is particularly good for confining conserved currents and conservation-law characteristics to the solution set of the system under consideration. Given a system in extended Kovalevskaya form, it is easy to derive expressions for all principal derivatives in terms of parametric derivatives since no nontrivial differential consequences arise. Conserved currents and characteristics depending only on the independent variables and parametric derivatives are called \emph{reduced conserved currents} and \emph{reduced characteristics}, respectively. It is obvious that for a general system of differential equations~$\mathcal L$, each tuple of differential functions that coincides with a conserved current of~$\mathcal L$ on solutions of~$\mathcal L$ is also a conserved current of~$\mathcal L$. Therefore, confining conserved currents to the solution set of the corresponding system is trivial. This is, however, not the case for conservation-law characteristics of general systems. At the same time, according to Lemma~3 in~\cite{mart79a} and the extended result in~\cite{popo2019b}, for a system in extended Kovalevskaya form, each conservation law admits a reduced characteristic, which allows one to neglect the triviality of conservation-law characteristics and gives a simple criterion on the triviality of conserved currents: A conserved current is trivial if and only if the corresponding reduced characteristic vanishes.
Moreover, given a system of partial differential equations $\mathcal L$ in the extended Kovalevskaya form~\eqref{Eq:ExtendedCauchyKovalevskayaForm}
and a reduced conserved current~$F=(F^1,\dots,F^n)$ of this system (and thus the corresponding reduced density is~$F^n$),
the components of the associated reduced characteristic are defined by (\cite{popo2019b}, see also~\cite[Proposition 7.41]{tsuj82a})
\[
\lambda^a=0\ \ \mbox{if}\ \ r_a=0 \quad\mbox{and}\quad
\lambda^a=\sum_{\beta\colon\beta_n=r_a-1}(-D)^{\beta-(r_a-1)\delta_n}\frac{\p F^n}{\p u^a_\beta}
=\mathsf E^{a,(r_a-1)\delta_n}F^n\ \ \mbox{if}\ \ r_a\geqslant1.
\]
Here $\mathsf E^{a,\alpha}$ is the higher-order Euler operator that corresponds to the derivative~$u^a_\alpha$
and which acts on an arbitrary differential function $P[u]$ according to
\[
\mathsf E^{a,\alpha}P=\sum_{\beta\geqslant\alpha}\frac{\beta!}{\alpha!(\beta-\alpha)!}(-D)^{\beta-\alpha}\frac{\p P}{\p u^a_\beta}.
\]
Recall also that the condition $\beta\geqslant\alpha$ for the multi-indices $\alpha=(\alpha_1,\dots,\alpha_n)$ and $\beta=(\beta_1,\dots,\beta_n)$
means that $\beta_1\geqslant\alpha_1$, \dots, $\beta_n\geqslant\alpha_n$,
and $\alpha!:=\alpha_1!\cdots\alpha_n!$ for any multi-index $\alpha$.
The existence of the explicit relation between reduced densities of a conservation law and its reduced characteristic
allows one to reformulate the above criterion on the triviality of conserved currents in terms of densities.

\medskip\par\noindent
{\bf Statement of the problem.}
In view of the above discussion, conservation-law characteristics are the better data compared to conserved currents. This is why, the empiric formulation of the inverse problem on conservation laws, which is given in the introduction, can be improved to:

\begin{problem*}
Derive the form of systems of differential equations with a prescribed set of con\-ser\-va\-tion-law characteristics.
\end{problem*}

In the most general setting, this formulation might still not be well posed and therefore it is not absolutely rigorous. This is why additional restrictions for the objects involved are necessary, e.g., on the form of systems, on the characteristic order or on the number of independent variables.

In particular, if we consider conservation-law characteristics of any order, then we should restrict the allowed form of systems in the statement of the problem by considering only systems in the extended Kovalevskaya form. In this case, one more proper kind of data, namely densities, exists. The corresponding specific formulation of the inverse problem on conservation laws is the following:

\begin{problem*}
Given a class of systems of differential equations in the extended Kovalevskaya form, find its subclass of systems admitting a prescribed set of reduced conservation-law characteristics (resp.\ reduced densities).
\end{problem*}

It often happens in practical applications, that the most important conservation laws possess low-order, or even zeroth-order, characteristics. For such characteristics the confining to the solution sets of the corresponding systems is trivial. Hence, it becomes inessential whether systems are of the extended Kovalevskaya form, cf.\ Section~\ref{sec:ConservativeParameterizationVorticityEquation}. We then have the following particular formulation of the inverse problem on conservation laws:

\begin{problem*}
Derive the form of systems of differential equations with prescribed set of low-order (in particular zeroth-order) conservation-law characteristics.
\end{problem*}

One more specific case is associated with the restriction of the number of independent variables to one, i.e., with ordinary differential equations. The corresponding inverse problem is presented in Section~\ref{sec:InverseProblemCLODE}.

The above particular rigorous formulations cannot cover all situations that arise in applications. For example, a class may consist of abnormal systems such as Maxwell's equations or Einstein field equations, which do not possess a representation in the extended Kovalevskaya form. Sometimes higher-order characteristics may arise and thus should be required to be admitted by differential equations. The formulation of the inverse problem on conservation laws in such situations needs more accurate investigations in order to guarantee that it is well posed.

\medskip\par\noindent
{\bf Tools for solution.}
Basic tools in the study of conservation laws (both in the formulation of the theory as well as in practical computations) are the `integration by parts'~\cite[p.~266]{olve93Ay}, which is just a specific version of the product rule, and its extension to the Lagrange identity, see e.g.~\cite[p.~67]{hart02a} for the case of linear ordinary differential operators. In particular, these tools are used in the definition of conservation-law characteristics~\cite[p.~266]{olve93Ay} and the proof of the fundamental theorem on correspondence between characteristics and conserved currents~\cite[Theorem 4.26]{olve93Ay}. This is why it is natural that these tools should also be applied for solving particular inverse problems on conservation laws.

For convenience we present here the form of the Lagrange identity for differential operators in total derivatives. Let $\mathsf P$ be a linear differential operator in total derivatives, $\mathsf P=\psi^\alpha[u]\mathrm D^\alpha$, where $\mathrm D^\alpha=\mathrm D^{\alpha_1}\cdots\mathrm D^{\alpha_n}$ and only a finite number of the coefficients~$\psi^\alpha$ are nonzero. Denote by~$\mathsf P^\dag$ its formally adjoint operator, $\mathsf P^\dag f=(-\mathrm D)^\alpha (\psi^\alpha f)$ for any differential function $f=f[u]$. The \emph{Lagrange identity} (also called \emph{generalized Green's formula} \cite[Section~12]{zhar92a})
implies that for any differential functions~$f$ and~$g$ of $u$,
 \begin{equation}\label{eq:LagrangeIdentity}
  f\,\mathsf Pg-g\,\mathsf P^\dag f\in\mathop{\rm im}\mathop{\rm Div},
 \end{equation}
where $\mathop{\rm im}\mathop{\rm Div}$ denotes the image of the total divergence operator.
In other words, there is a tuple of differential functions~$F=(F^1,\dots, F^n)$ such that $f\,\mathsf Pg-g\,\mathsf P^\dag f=\mathop{\rm Div} F$.
Moreover, the tuple~$F$ can be explicitly represented in terms of total derivatives of $f$, $g$ and $\psi^\alpha$ \cite[Proposition~A.4]{zhar92a}.

Another tool for the solution of the inverse problem on conservation laws is provided by equation~\eqref{eq:DeterminingEquationsCharacteristics}. If $\lambda$ is a prescribed characteristic, then equation~\eqref{eq:DeterminingEquationsCharacteristics} implies a system of defining equations for~$L$. This tool is efficient when the number of prescribed characteristics is sufficiently large, cf.\ Section~\ref{sec:ConservativeParameterizationWithDirectMethods}.

\section[The inverse problem on first integrals for ordinary differential equations]
{The inverse problem on first integrals\\ for ordinary differential equations}\label{sec:InverseProblemCLODE}

\looseness=-1
The case of ordinary differential equations, where the number of independent variables equals one, is quite specific from various points of view including conservation laws. This is why particular terminology is used for related notions. Thus, conserved currents which are singletons and analogous to densities in this case, are called constants of the motion or first integrals and the associated characteristics are called integrating factors.
Given a system~$\mathcal L$ of ordinary differential equations that admits a representation in the Kovalevskaya form,%
\footnote{
For systems of ordinary differential equations, the Kovalevskaya form is called the \emph{canonical form}.
The particular case of the canonical form, where all equations are of the first order, is called the normal Cauchy form or, shortly, the normal form.
}
it is natural to assume that its first integrals involve only derivatives of orders less than orders of the corresponding leading derivatives
and, if nontrivial, necessarily depend on subleading derivatives.
Since in the case of one independent variable null divergences are exhausted by constants, the equivalence of conserved currents then degenerates for the system~$\mathcal L$ to adding arbitrary constants to first integrals and becomes quite inessential.
As a result, the inverse problem on conservation laws in the class of systems of ordinary differential equations reducing to the canonical form can be interpreted as the inverse problem on first integrals,
which is stated in the above particular terminology as follows:

\begin{problem*}
Find the general form of systems of ordinary differential equations that admit a prescribed set of integrating factors.
\end{problem*}

Here we consider the simplest case of a single ordinary differential equation for a single unknown function
since in this case the solution of the inverse problem on first integrals possesses an especially nice representation.
Let
\[
 L[u]=0
\]
be a single ordinary differential equation in the independent variable~$x_1=:t$ and the single dependent variable~$u$. 
We aim to find the functional form of~$L$ that admits~$p$ first integrals $I^1$,~\dots,~$I^p$, 
where the associated integrating factors $\lambda^1$,~\dots,~$\lambda^p$ are totally linearly independent.

The case of $p=1$, i.e., choosing a single integrating factor, is trivial. From the characteristic form $\mathrm D_t I^1=\lambda^1 L$ for the first integral $I^1$ corresponding to the integrating factor~$\lambda^1$, we can resolve $L=\mathrm D_tI^1/\lambda^1$ provided that $\lambda^1\ne0$. This formula, where $I^1$ runs through the set of differential functions, gives the general form of the left hand sides of single ordinary differential equations that admit $\lambda^1$ as an integrating factor.

For more than one first integral, deriving corresponding formulas is much more involved.
In particular, we need to extend the notion of Darboux transformation~\cite[p.~9]{matv91a} to differential functions.
By $\mathrm W(f^1,\dots,f^k)$ we denote the Wronskian of differential functions~$f^1[u]$, \dots, $f^k[u]$ in the total derivative~$\mathrm D_t$, i.e.,
$\mathrm W(f^1,\dots,f^k)=\det(\mathrm D_t^{l'-1}f^l)_{l,l'=1,\dots,k}$.

\begin{definition}
Given totally linearly independent differential functions~$f^1[u]$, \dots, $f^k[u]$, 
whose Wronskian does not vanish, the differential function
 \[
  \mathrm{DT}[f^1,\dots,f^k]G=\frac{\mathrm W(f^1,\dots,f^k,G)}{\mathrm W(f^1,\dots,f^k)},
 \]
is called the \emph{Darboux transformation} of a differential function~$G[u]$ with respect to~$f^1$, \dots, $f^k$.
\end{definition}

The Darboux transformation $\mathrm{DT}[f^1,\dots,f^k]G$ is the result of the action of a differential operator in the total derivative~$\mathrm D_t$ on the differential function~$G$.
We call this operator the \emph{Darboux operator} with respect to~$f^1$, \dots, $f^k$ and denote it by~$\mathrm{DT}[f^1,\dots,f^k]$. Its formally adjoint operator is denoted by~$\mathrm{DT}[f^1,\dots,f^k]^\dag$.

The following theorem describes the form of a single ordinary differential equation admitting $p$ integrating factors with nonvanishing Wronskian.

\begin{theorem}\label{thm:OnFormOfODE}
A single ordinary differential equation~$L[u]=0$ admits 
$p$ integrating factors $\lambda^1$,~\dots,~$\lambda^p$ with $\mathrm W(\lambda^1,\dots,\lambda^p)\ne0$
if and only if the left hand side~$L$ is of the form
\begin{equation}\label{eq:ODEFormInverseCL}
L=\mathrm{DT}[\lambda^1,\dots,\lambda^p]^\dag H,
\end{equation}
where~$H$ is a differential function of $u$.
\end{theorem}

\begin{proof}
Consider an arbitrary but fixed ordinary differential equation~$L=0$ with integrating factors~$\lambda^1$, \dots, $\lambda^p$. 
Denote by~$I^s$ a first integral corresponding to $\lambda^s$, $s=1,\dots,p$. 
We introduce the tuple of differential functions $(\varphi^1,\dots,\varphi^p)$ 
that is adjoint to $(\lambda^1,\dots,\lambda^p)$ and whose components are defined by
\begin{equation}\label{eq:DefinitionAdjointFunctions}
 \varphi^s=(-1)^{p-s}\frac{\mathrm W(\lambda^1,\dots,\cancel{\lambda^s},\dots,\lambda^p)}{\mathrm W(\lambda^1,\dots,\lambda^p)},
\end{equation}
where $(\lambda^1,\dots,\cancel{\lambda^s},\dots,\lambda^p)$ for a fixed~$s$ denotes the tuple obtained from the tuple~$(\lambda^1,\dots,\lambda^p)$
by excluding~$\lambda^s$.
Conversely, for $(\lambda^1,\dots,\lambda^p)$ we can write in terms of $(\varphi^1,\dots,\varphi^p)$,
\[
 \lambda^s=(-1)^{s-1}\frac{\mathrm W(\varphi^1,\dots,\cancel{\varphi^s},\dots,\varphi^p)}{\mathrm W(\varphi^1,\dots,\varphi^p)}.
\]
Note that we have
\begin{equation}\label{eq:DerivativePropertiesAdjointFunctions}
 (\mathrm D_t^k\varphi^s)\lambda^s=\left\{\begin{array}{cl} 0 &\mbox{if}\quad 0\leqslant k<p-1,\\ (-1)^{p-1} &\mbox{if}\quad k=p-1.\end{array}\right.
\end{equation}
We set $H:=-\varphi^sI^s$.
Totally differentiating this equality and using $\mathrm D_tI^s=\lambda^sL$ and~\eqref{eq:DerivativePropertiesAdjointFunctions}, we get
\begin{gather*}
-\mathrm D_tH=(\mathrm D_t\varphi^s)I^s+\varphi^s\lambda^sL=(\mathrm D_t\varphi^s)I^s,\\
-\mathrm D_t^2H=(\mathrm D_t^2\varphi^s)I^s+(D_t\varphi^s)\lambda^sL=(\mathrm D_t^2\varphi^s)I^s,\quad \dots\ .
\end{gather*}
Iterating this procedure, we derive overall
\begin{align}\label{eq:SystemForFandL}
\begin{split}
 &-\mathrm D_t^kH=(\mathrm D_t^k\varphi^s)I^s,\quad k=0,\dots,p-1,\\
 &-\mathrm D_t^pH=(\mathrm D_t^p\varphi^s)I^s+(-1)^{p-1}L.
\end{split}
\end{align}
We can interpret the above system as a system of linear algebraic equations with respect to~$I^s$ and~$L$. Employing Cramer's rule we can solve the above linear system to obtain
\[
 L=(-1)^p\frac{\mathrm W(\varphi^1,\dots,\varphi^p,H)}{\mathrm W(\varphi^1,\dots,\varphi^p)}= (-1)^p\mathrm{DT}[\varphi^1,\dots,\varphi^p]H,
\]
which, in view of the property
$ \mathrm{DT}[\varphi^1,\dots,\varphi^p]=(-1)^p\mathrm{DT}[\lambda^1,\dots,\lambda^p]^\dag$,
gives the form~\eqref{eq:ODEFormInverseCL}.

In turn, if the representation~\eqref{eq:ODEFormInverseCL} holds then we find from multiplying~\eqref{eq:ODEFormInverseCL} with the integrating factor~$\lambda^s$ that
 \begin{equation}\label{eq:ProofODEFormInverseCL}
 \lambda^s L=\lambda^s\mathrm{DT}[\lambda^1,\dots,\lambda^p]^\dag H=\lambda^s\mathrm{DT}[\lambda^1,\dots,\lambda^p]^\dag H-H\mathrm{DT}[\lambda^1,\dots,\lambda^p]\lambda^s,
 \end{equation}
where we subtracted zero in the last expression since $\mathrm{DT}[\lambda^1,\dots,\lambda^p]\lambda^s=0$. In view of the Lagrange identity~\eqref{eq:LagrangeIdentity} the right hand side of~\eqref{eq:ProofODEFormInverseCL} is the total derivative of some differential function.
This means that  for each~$s$ the differential function~$\lambda^s$ is an integrating factor of the equation $L=0$.
\end{proof}

\begin{remark}\label{rem:HInFormOfODE}
Since $\ker\mathrm{DT}[\lambda^1,\dots,\lambda^p]^\dag=\mathrm{DT}[\varphi^1,\dots,\varphi^p]=\langle\varphi^1,\dots,\varphi^p\rangle$, 
it follows from the proof of Theorem~\ref{thm:OnFormOfODE} 
that the differential function~$H$ in the representation~\eqref{eq:ODEFormInverseCL} for a given differential function~$L$ is necessarily of the form 
\[
H=-\varphi^sI^s+c_s\varphi^s,
\]
where the tuple of differential functions $(\varphi^1,\dots,\varphi^p)$ is adjoint to~$(\lambda^1,\dots,\lambda^p)$ 
with components $\varphi^s$ defined by~\eqref{eq:DefinitionAdjointFunctions},
$I^s$ is a first integral of the equation $L=0$ 
that are associated with the integrating factor~$\lambda^s$, $\mathrm D_tI^s=\lambda^sL$, $s=1,\dots,p$, 
and $c_1$, \dots, $c_p$ are arbitrary constants. 
The indeterminacy in~$H$ is induced by the indeterminacy in defining first integrals by associated integrating factors, 
$\mathrm D_t\tilde I^s=\lambda^sL$ as well if and only if $\tilde I^s=I^s+c_s$ for a constant~$c_s$.
\end{remark}

\begin{corollary}
If a differential function~$L$ admits the representation~\eqref{eq:ODEFormInverseCL},
then the ordinary differential equation~$L=0$ possesses~$p$ totally linearly independent first integrals
\[
I^s=(-1)^{p-s+1}\frac{\mathrm W(\varphi^1,\dots,\cancel{\varphi^s},\dots,\varphi^p,H)}{\mathrm W(\varphi^1,\dots,\varphi^p)}
\]
respectively associated with the integrating factors~$\lambda^s$,
where the tuple of differential functions $(\varphi^1,\dots,\varphi^p)$ is adjoint to~$(\lambda^1,\dots,\lambda^p)$ 
with components $\varphi^s$ defined by~\eqref{eq:DefinitionAdjointFunctions}.
\end{corollary}

\begin{proof}
The expression for~$I^s$ is obtained from the system~\eqref{eq:SystemForFandL} using Cramer's rule.
\end{proof}

We now study a few properties of the representation~\eqref{eq:ODEFormInverseCL}.

\begin{proposition}\label{pro:RepresentationODESpanBasis}
 The representation~\eqref{eq:ODEFormInverseCL} depends rather on the linear span $\langle\lambda^1,\dots,\lambda^p\rangle$, 
 not on the fixed basis $(\lambda^1,\dots,\lambda^p)$ itself.
\end{proposition}

\begin{proof}
Consider $\tilde\lambda^s=\lambda^{s'}a_{s's}$, where $A=(a_{s's})$ is a nondegenerate constant matrix. Since $\mathrm W(\tilde\lambda^1,\dots,\tilde\lambda^p)=\mathrm W(\lambda^1,\dots,\lambda^p)\det A$ and $\mathrm W(\tilde\lambda^1,\dots,\tilde\lambda^p,H)=\mathrm W(\lambda^1,\dots,\lambda^p,H)\det A$, it follows from the definition of Darboux operator
that $\mathrm{DT}(\tilde\lambda^1,\dots,\tilde\lambda^p)= \mathrm{DT}(\lambda^1,\dots,\lambda^p)$. 
The linear combining of integrating factors $\lambda^1$,~\dots,~$\lambda^p$ 
is underlaid by the analogous linear combining of the associated first integrals $I^1$,~\dots,~$I^p$, 
$\tilde I^s=I^{s'}a_{s's}$.
The components of the tuple $(\tilde\varphi^1,\dots,\tilde\varphi^p)$ adjoint to~$(\tilde\lambda^1,\dots,\tilde\lambda^p)$ 
are given by 
\[
\tilde\varphi^s=(-1)^{p-s}\sum_{s'=1}^p\frac{\mathrm W(\sum_{\sigma'\ne s'}\lambda^{\sigma'}a_{\sigma'\sigma},\sigma\ne s)}{\mathrm W(\lambda^1,\dots,\lambda^p)\det A}
=\frac{(-1)^{s+s'}}{\det A}\sum_{s'=1}^p\varphi^{s'}A_{s's},
\]
where $A_{s's}$ is the $(s',s)$ minor of~$A$, and $\sigma$ and $\sigma'$ run from 1 to~$p$.
Therefore, the differential function $H:=-\varphi^sI^s$ is preserved in the course of this linear combining, 
$H=\tilde H:=-\tilde\varphi^s\tilde I^s$.
\end{proof}

\begin{remark}\label{rem:FormOfODEGaugeTrans}
Linear combinations of first integrals of an ordinary differential equation 
present the simplest particular case of general functional combinations of such integrals. 
The latter combinations are first integrals of the same equation as well. 
Given an ordinary differential equation $L[u]=0$ and its $p$ first integrals $I^1$,~\dots,~$I^p$ such that 
the Wronskian of the corresponding integrating factors $\lambda^1$,~\dots,~$\lambda^p$ does not vanish, 
consider functional combinations $\tilde I^s:=F^s(I^1,\dots,I^p)$ with nonzero Jacobian $\mathrm J:=\det(F^s_{I^{s'}})$. 
The integrating factor for the first integral $\tilde I^s$ is $\tilde\lambda^s=F^s_{I^{s'}}\lambda^{s'}$, 
and on solutions of the equation $L[u]=0$ we get 
$\mathrm W(\tilde\lambda^1,\dots,\tilde\lambda^p)=\mathrm W(\lambda^1,\dots,\lambda^p)\mathrm J\ne0$. 
Setting $\tilde H:=-\tilde\varphi^s\tilde I^s$, 
where $(\tilde\varphi^1,\dots,\tilde\varphi^p)$ is the tuple adjoint to~$(\tilde\lambda^1,\dots,\tilde\lambda^p)$, 
we derive the representation~\eqref{eq:ODEFormInverseCL} 
in terms of the transformed integrating factors $\tilde\lambda^1$,~\dots,~$\tilde\lambda^p$. 
In other words, the representation~\eqref{eq:ODEFormInverseCL} is invariant with respect to 
the prolongations of gauge transformations of functional re-combining first integrals to integrating factors.
\end{remark}

\begin{proposition}\label{pro:RepresentationODEUndeterminedness}
 For fixed~$L$ and fixed $\langle\lambda^1,\dots,\lambda^p\rangle$, the general form of~$H$ is $H=H^0+c_s\varphi^s$,
 where $H^0$ is a particular value of~$H$, and $c_1$, \dots, $c_p$ are arbitrary constants.
\end{proposition}

\begin{proof}
 It suffices to note that the linear span $\langle\varphi^1,\dots,\varphi^p\rangle$ coincides with the kernel of the Darboux operator $\mathrm{DT}(\varphi^1,\dots,\varphi^p)=(-1)^p\mathrm{DT}(\lambda^1,\dots,\lambda^p)^\dag$.
\end{proof}

\begin{proposition}\label{pro:PointTransOfODERepresentation}
Given a differential function~$L[u]$ represented in the form~\eqref{eq:ODEFormInverseCL}, 
a nonvanishing differential function~$\Lambda[u]$ and 
a point transformation~$\mathcal T$: $\tilde t=T(t,u)$, $\tilde u=U(t,u)$, 
the differential function~$\tilde L[\tilde u]$ of $\tilde u=\tilde u(\tilde t)$
that is defined by $\tilde L[\tilde u]=\Lambda[u]L[u]$
admits the counterpart of the representation~\eqref{eq:ODEFormInverseCL},
$\tilde L=\mathrm{DT}[\tilde\lambda^1,\dots,\tilde\lambda^p]^\dag\tilde H$, with
differential functions $\tilde\lambda^1[\tilde u]$, \dots, $\tilde\lambda^p[\tilde u]$, $H[\tilde u]$ 
obtained from the conditions 
\[
\tilde\lambda^s[\tilde u]=\frac{\lambda^s[u]}{\Lambda[u]\,\mathrm D_tT(t,u)},\quad
\tilde H[\tilde u]=\Lambda[u]\big(\mathrm D_tT(t,u)\big)^pH[u].
\]
\end{proposition}

\begin{proof}
If the differential function~$L[u]$ is represented in the form~\eqref{eq:ODEFormInverseCL},
then the differential functions $\lambda^1[u]$, \dots, $\lambda^p[u]$ are integrating factors of the equation $L[u]=0$.
Let $I^1[u]$, \dots, $I^p[u]$ be associated first integrals of this equation.
For each~$s$, the point transformation~$\mathcal T$ maps $I^s$ and~$\lambda^s$
to a first integral~$\tilde I^s$ and the corresponding integrating factor~$\tilde\lambda^s$
of the equation $\tilde L[\tilde u]=0$.
Therefore, we have the counterpart of the representation~\eqref{eq:ODEFormInverseCL} for the function~$\tilde L$.
Since $\tilde I^s=I^s$, $\tilde L=\Lambda L$ and
$\tilde\lambda^s\tilde L=\mathrm D_{\tilde t}\tilde I^s=(\mathrm D_tI^s)/(\mathrm D_tT)=\lambda^sL/(\mathrm D_tT)$,
the prolongation of~$\mathcal T$ to~$\lambda^s$ is $\tilde\lambda^s=\lambda^s/(\Lambda\mathrm D_tT)$.
Wronskians have the following simple properties:
$\mathrm W(gf^1,\dots,gf^k)=g^k\mathrm W(f^1,\dots,f^k)$ and
$\tilde{\mathrm W}(f^1,\dots,f^k)=(\mathrm D_tT)^{k-1}\mathrm W(f^1,\dots,f^k)$
for arbitrary differential functions~$f^1[u]$, \dots, $f^k[u]$ and $g[u]$,
where $\tilde{\mathrm W}$ denote the Wronskian in the total derivative~$\mathrm D_{\tilde t}$ 
with the independent variable $\tilde t=T(t,u)$ and the dependent variable $\tilde u=U(t,u)$.
This is why the prolongation of~$\mathcal T$ to~$\varphi^s$ is $\tilde\varphi^s=\Lambda(\mathrm D_tT)^p\varphi^s$.
Then the expressions for~$H$ and~$\tilde H$, $H=-\varphi^sI^s$ and $\tilde H=-\tilde\varphi^s\tilde I^s$,
imply the prolongation of~$\mathcal T$ to~$H$.
\end{proof}

Consider a linear $r$th order ordinary differential equation $L=0$ with leading coefficient~1. Then the Lagrange identity implies that this equation admits~$r$ totally linearly independent integrating factors~$\lambda^s$ that are functions of~$t$ only and satisfy the adjoint of the equation $L=0$. Hence, the Darboux operator $\mathrm{DT}(\lambda^1,\dots,\lambda^r)^\dag$ coincides with the operator associated with the equation $L=0$. This observation leads to the following proposition.

\begin{proposition}\label{pro:OnLinearEquations}
 Let $L=0$ be a linear $r$th order ordinary differential equation with leading coefficient~1. 
 For the representation~\eqref{eq:ODEFormInverseCL} 
 constructed with~$p=r$ totally linearly independent integrating factors depending only on~$t$, 
 the adjoint Darboux operator coincides with the differential operator associated with the equation, 
 and $H=u$ modulo adding an arbitrary solution of the equation. 
\end{proposition}

It is obvious that each linear ordinary differential equation also admits integrating factors that depend on derivatives of $u$.
The representation~\eqref{eq:ODEFormInverseCL} constructed with using such integrating factors involves, even for linear equations,
expressions nonlinear in derivatives of $u$.

In the following we will consider a few simple examples of well-known ordinary differential equations  
and derive the representation~\eqref{eq:ODEFormInverseCL} for the right-hand side of each of these equations. 
Varying the involved differential function~$H$ will result in the representation for all ordinary differential equations  
that admit the same first integrals as the original equation does
but are, in general, nonlinear even if the original equation is linear. 
These examples hint at ways of studying the inverse problem on conservation laws 
for the more complicated case of nonlinear partial differential equations,  
including the construction of ansatzes for conservative parameterizations.

\begin{example}\label{ex:ElementarySecondOrderODE}
Consider the elementary second-order ODE $L:=u''=0$. 
This equation admits the integrating factors $\lambda^1=1$ and $\lambda^2=t$, 
with the associated first integrals $I^1=u'$ and $I^2=tu'-u$. 
Using~\eqref{eq:DefinitionAdjointFunctions}, we determine 
the adjoint functions~$\varphi^1$ and~$\varphi^2$ to be $\varphi^1=-t$ and $\varphi^2=1$. 
Now evaluating the representation~\eqref{eq:ODEFormInverseCL}, we obtain
\[
 L=\mathrm{DT}[\lambda^1,\lambda^2]^\dag H=(-1)^2\mathrm{DT}[\varphi^1,\varphi^2] H=\frac{\mathrm W(\varphi^1,\varphi^2,H)}{\mathrm W(\varphi^1,\varphi^2)}
 =\mathrm D_t^2H.
\]
Choosing the particular value $H=-\varphi^1I^1-\varphi^2I^2$ for the differential function~$H$ according Remark~\ref{rem:HInFormOfODE}, 
we find $H=u$, which agrees with Proposition~\ref{pro:OnLinearEquations}. 
Varying~$H$ in the set of differential functions of~$u$ with $\ord\mathrm D_t^2H>0$, we get the following obvious claim. 
An ordinary differential equation $L[u]=0$ possesses the integrating factors~1 and~$t$ 
if and only if $L=\mathrm D_t^2H$ for some~$H[u]$. 
Similar claims hold for all tuples of~$\lambda$'s considered below in the present and the successive examples 
but we omit these claims for brevity.

At the same time, the elementary second-order ODE also admits the integrating factor~$u'$ associated  with the first integral $(u')^2/2$.
In the representation~\eqref{eq:ODEFormInverseCL} for the equation $u''=0$ we can use, e.g., the pairs $(\lambda^1,\lambda^2)=(1,u')$ or $(\lambda^1,\lambda^2)=(t,u')$.
In the first case, the adjoint functions are $\varphi^1=-u'/u''$ and $\varphi^2=1/u''$ and hence we obtain the representation
\[
 L= \mathrm D_t\left(\mathrm D_t+\frac{u'''}{u''}\right)H,\quad \textup{where}\quad H=\frac{(u')^2}{2u''}.
\]
In the second case we compute the adjoint functions $\varphi^1=-u'/(tu''-u')$ and $\varphi^2=t/(tu''-u')$. The representation~\eqref{eq:ODEFormInverseCL} then becomes
\[
 L=\left(\mathrm D_t^2+\frac{tu'''}{tu''-u'}\mathrm D_t+ \frac{2u'''+tu''''}{tu''-u'}-\frac{(tu''')^2}{(tu''-u')^2}\right)H,\quad \textup{where}\quad H=\frac{tu'^2-2uu'}{2(tu''-u')}.
\]

Moreover, one can derive the representation~\eqref{eq:ODEFormInverseCL} for $u''=0$ using three integrating factors. Setting $\lambda^1=1$, $\lambda^2=t$ and $\lambda^3=u'$, the associated adjoint functions are $\varphi^1=(tu''-u')/u'''$, $\varphi^2=-u''/u'''$ and $\varphi^3=1/u'''$. In this case, the representation~\eqref{eq:ODEFormInverseCL} reads
\begin{align*}
 L=-\mathrm D_t^2\left(\mathrm D_t+\frac{u^{\rm iv}}{u'''}\right)H=-\mathrm D_t^2\frac{\mathrm D_t(u'''H)}{u'''}
 \quad\mbox{with}\quad H=-\frac{uu''}{u'''}+\frac{(u')^2}{2u'''}.
\end{align*}

Note that the last three representations involve derivatives of higher order than the order of the original equation. 
These derivatives are canceled after the substitution of an appropriate value of the differential parameter function~$H$ 
in the representation and the subsequent expansion. 

More generally, any first integral of the equation $u''=0$ is a function~$G$ of the two most elementary integrals $u'$ and $\omega=tu'-u$. 
The associated integrating factor is $\lambda=G_{u'}+tG_\omega$. 
Thus we could construct further representations of the form~\eqref{eq:ODEFormInverseCL} 
using an arbitrary number of totally linearly independent integrating factors of this equation.

This example illustrates that even starting from a simple ordinary differential equation with simple integrating factors 
one can obtain quite cumbersome representations of the form~\eqref{eq:ODEFormInverseCL}.
\end{example}

\begin{remark}
Example~\ref{ex:ElementarySecondOrderODE} shows
that the representation~\eqref{eq:ODEFormInverseCL} may involve derivatives of~$u$ whose order is higher than~$\ord L$.
Moreover, the representation may be singular for any solution of the corresponding equation $L[u]=0$.
In general, this singularity always occurs when the corresponding first integrals are functionally dependent.
Indeed, suppose that $I^p=\Theta(I^1,\dots, I^{p-1})$.
Totally differentiating this equality with respect to~$t$, we derive
$\lambda^p=\Theta_{I^1}\lambda^1+\dots+\Theta_{I^{p-1}}\lambda^{p-1}$
and thus
$\mathrm D_t^k\lambda^p=\Theta_{I^1}\mathrm D_t^k\lambda^1+\dots+\Theta_{I^{p-1}}\mathrm D_t^k\lambda^{p-1}+R^k$, $k=1,2,\dots$,
where the differential function~$R^k=R^k[u]$ vanishes on each solution of the equation $L[u]=0$.
Therefore, the Wronskian $\mathrm W(\lambda^1,\dots,\lambda^p)$ has the same property.
\end{remark}

\begin{example}
 As a second initial equation, consider the equation for the classical harmonic oscillator $L:=u''+u=0$. The harmonic oscillator possesses the two integrating factors $\lambda^1=-\sin t$ and $\lambda^2=\cos t$, with the associated first integrals $I^1=u\cos t-u'\sin t$ and $I^2=u\sin t+u'\cos t$. The adjoint functions to the integrating factors are $\varphi^1=-\cos t$ and $\varphi^2=-\sin t$. We again evaluate the representation~\eqref{eq:ODEFormInverseCL} and find
\[
 L=\mathrm{DT}[\lambda^1,\lambda^2]^\dag H=(-1)^2\mathrm{DT}[\varphi^1,\varphi^2] H=\frac{\mathrm W(\varphi^1,\varphi^2,H)}{\mathrm W(\varphi^1,\varphi^2)}=\mathrm D_t^2H+H.
\]
In accordance with Proposition~\ref{pro:OnLinearEquations}, $H=-\varphi^1I^1-\varphi^2I^2=u$ is associated with the initial equation, 
and the other values of~$H$ with $\ord(\mathrm D_t^2H+H)>0$ 
leads to other ordinary differential equations 
admitting the same integrating factors $\lambda^1=-\sin t$ and $\lambda^2=\cos t$

For the harmonic oscillator we also have the integrating factor~$\lambda=u'$ 
and the most general first integral is a function of $I^1$ and $I^2$. 
But the consideration of the representation~\eqref{eq:ODEFormInverseCL} originated from the harmonic oscillator is not needed. 
In fact, for second-order linear ODEs it is sufficient to study only the elementary equation 
since any such equation is similar to the elementary equation 
with respect to a foliation-preserving point transformation that is linear in~$u$. 
This is why any representation of the form~\eqref{eq:ODEFormInverseCL} with integrating factors of a second-order linear ODE 
is the image, with respect to such a transformation, of a similar representation with integrating factors of the elementary equation.
\end{example}

\begin{example}\label{ex:Lorenz1963}
To give an example with a nonlinear initial equation, we now consider the famous Lorenz 1963 model. 
This dynamical system of three equations reads
\begin{equation}\label{eq:LorenzSystem}
  \frac{\mathrm dx}{\mathrm dt}=\sigma y -m\sigma x, \quad \frac{\mathrm dy}{\mathrm dt}=x(r-z)-my,\quad \frac{\mathrm dz}{\mathrm dt}=xy-mbz,
\end{equation}
where it is conventional to denote the dependent variables of the system 
by $x$, $y$ and $z$, and $\sigma$, $r$ and $b$ are non-dimensional constants. 
The control parameter $m$ governs the strength of the dissipation. 
In the original Lorenz system, $m=1$. In the case of $m=0$, system~\eqref{eq:LorenzSystem} 
is called the conservative Lorenz system, and admits two first integrals, which are
\[
 \tilde I^1=\frac12x^2-\sigma z,\quad \tilde I^2=\frac12(y^2+z^2)-rz,
\]
see e.g.~\cite{bihl11Ay}. 
To fit the example of the conservative Lorenz system in the framework developed above, 
we convert~\eqref{eq:LorenzSystem} for the case of $m=0$ into a single ordinary differential equation. 
Suppose that $x$ is not a constant. 
Then expressing~$y$ from the first equation and~$z$ from the second equation in view of the expression for~$y$ 
and substituting into the last equation leads to the following third-order equation for~$x$:
\begin{equation}\label{eq:LorenzSystemThirdOrderEquation}
 L:=\left(\frac{x''}x\right)'+xx'=0.
\end{equation}
The two first integrals for~\eqref{eq:LorenzSystemThirdOrderEquation} 
corresponding to the above first integrals for the conservative Lorenz system, up to constant multipliers, are
\[
 I^1=\frac{x^2}2+\frac{x''}x,\quad I^2=\frac{(x')^2}2+\frac12\left(\frac{x''}x\right)^2,
\]
with the associated integrating factors~$\lambda^1=1$ and $\lambda^2=x''/x$. 
The adjoint functions for $(\lambda^1,\lambda^2)$ are $\varphi^1=-(x''/x)/(x''/x)'$ and $\varphi^2=1/(x''/x)'$ 
and the representation~\eqref{eq:ODEFormInverseCL} for the equation~\eqref{eq:LorenzSystemThirdOrderEquation} reads
\[
 L=\mathrm D_t\left(\mathrm D_t+\frac{(x''/x)''}{(x''/x)'}\right)H \quad\mbox{with}\quad
 H=\frac{xx''-(x')^2+(x''/x)^2}{2(x''/x)'}.
\]
\end{example}

For practical computations and more accurate estimations of orders of involved differential functions, an alternative representation to~\eqref{eq:ODEFormInverseCL} is useful. We define $\hat\varphi^s=\mathrm W(\lambda^1,\dots,\lambda^p)\varphi^s$ and $\hat H=\mathrm W(\lambda^1,\dots,\lambda^p)H$. Substituting $\varphi^s=\hat\varphi^s/\mathrm W(\lambda^1,\dots,\lambda^p)$ and $H=\hat H/\mathrm W(\lambda^1,\dots,\lambda^p)$ in the representation~\eqref{eq:ODEFormInverseCL}, and using the Wronskian property $\mathrm W(gf^1,\dots,gf^k)=g^k\mathrm W(f^1,\dots,f^k)$, we obtain the alternative representation
\begin{equation}\label{eq:ODEFormInverseCLAlternative}
 L=(-1)^p\frac{\mathrm W(\hat\varphi^1,\dots,\hat\varphi^p,\hat H)}{(\mathrm W(\lambda^1,\dots,\lambda^p))^p}.
\end{equation}

\begin{lemma}\label{lem:OnOrderOfRepresentation}
 Let $\ord L=r$ and $q:=\max_s\ord\lambda^s$. Then, the differential function~$\hat H$ in the representation~\eqref{eq:ODEFormInverseCLAlternative} satisfies the condition $\ord\hat H\leqslant\max(r-p,q+p-2)$.
\end{lemma}

\begin{proof}
 By definition of~$q$, we have $q\in\{-\infty\}\cup2\mathbb N_0$, $\ord\lambda^s\leqslant q$ and hence $\ord\mathrm W(\lambda^1,\dots,\lambda^p)\leqslant q+p-1$. The definition of $\hat\varphi^s$ implies that $\ord\hat\varphi^s\leqslant q+p-2$ and thus $\ord\mathrm D_t^p\hat\varphi^s\leqslant q+2p-2$. As $\ord L=r$, from the alternative representation~\eqref{eq:ODEFormInverseCLAlternative} we obtain the estimate $\ord\mathrm D_t^p\hat H\leqslant\max(r,q+2p-2)$. Therefore, $\ord\hat H\leqslant\max(r-p,q+p-2)$, proving the assertion.
\end{proof}

\begin{corollary}
 If $q\leqslant r-2p+2$ then $\ord\hat H\leqslant r-p$.
\end{corollary}

\begin{proof}
The inequality $q\leqslant r-2p+2$ implies that $q+p-2\leqslant r-p$ and therefore the estimate for $\ord\hat H$ from Lemma~\ref{lem:OnOrderOfRepresentation} reduces to $\ord\hat H\leqslant r-p$.
\end{proof}

\begin{lemma}
Let $\ord\lambda^s\leqslant r-2p+1$ for all~$s$. Then in the representation~\eqref{eq:ODEFormInverseCLAlternative} $\ord L=r$ if and only if  $p\leqslant r$ and $\ord\hat H=r-p$.
\end{lemma}

\begin{proof}
Similarly to the above, we have $\ord\mathrm W(\lambda^1,\dots,\lambda^p)\leqslant r-p$, $\ord\hat\varphi^s\leqslant r-p-1$ and thus $\ord\mathrm D_t^p\hat\varphi^s\leqslant r-1$.

Suppose that $\ord L=r$.
If $p>r$, then $\ord\lambda^s=-\infty$ and hence necessarily $\ord\hat H\geqslant0$.
This implies that the order of the right hand side in~\eqref{eq:ODEFormInverseCLAlternative} is not less than~$p$, which is greater than the order of the left hand side~$L$.
The obtained contradiction means that $p\leqslant r$.
Then $\ord\mathrm D_t^p\hat H=r$, which is equivalent to $\ord\hat H=r-p$.

The converse assertion is obvious.
\end{proof}

\begin{corollary}
 If $\ord\lambda^s\leqslant r-2p$, then additionally in the representation~\eqref{eq:ODEFormInverseCL} $\ord H=\ord\hat H= r-p$.
\end{corollary}

\begin{remark}
The inverse problem on conservation laws in general becomes more complicated 
if additional restrictions, e.g., on the order of equations to be constructed are posed. 
The solution of such a modified problem may still involve representations like~\eqref{eq:ODEFormInverseCL} 
but may also need specific tools depending on the tuple of integrating factors under consideration.
Let us construct, for instance, ordinary differential equations, $L=0$, of order less than or equal to three
that admit 1, $t$ and $u'$ as integrating factors.
These differential functions are integrating factors of the elementary equation $u''=0$.
Following Example~\ref{ex:ElementarySecondOrderODE}, we obtain that
$L=\mathrm D_t^2H[u]$, where $\ord H\leqslant1$ and $u'\mathrm D_t^2H\in\mathop{\rm im}\nolimits\mathrm D_t$.
``Integrating by parts'' in the latter condition gives that $u'''H\in\mathop{\rm im}\nolimits\mathrm D_t$,
i.e., $\mathsf E(u'''H)=0$, where $\mathsf E$ is the Euler operator.
After expanding and splitting the equation $\mathsf E(u'''H)=0$ with respect to derivatives of~$u$,
we derive the system of determining equations for~$H$, which reduces to $H_{u'}=H_{uu}=H_{tu}=H_{ttt}=0$.
Therefore, the general form of equations to be constructed is $L=c_1u''+c_0$,
where $c_0$ and~$c_1$ are arbitrary constants with $c_1\ne0$.
Thus, we have the interesting phenomenon that there are no third-order ordinary differential equations
admitting 1, $t$ and $u'$ as integrating factors, which indicates a certain inconsistency among these integrating factors.
At the same time, similar second-order equations exist although they reduce to the elementary equation $u''=0$
by an obvious transformation.
\end{remark}

\section[The inverse problem on conservation laws for evolution equations]
{The inverse problem on conservation laws\\ for evolution equations}\label{sec:InverseProblemCLsEvolutionEquations}

Using the tool of Darboux transformations the results from the previous section on ordinary differential equations can be extended to the case of (1+1)-dimensional evolution equations. We present the analogous results to the ones given for ordinary differential equations. In many aspects, the theory of evolution equations is very close to the theory of ordinary differential equations and~$t$ plays the role of a parameter. A more extended discussion including the proofs of the following statements will be presented in~\cite{popo2019a}.

As with ordinary differential equations, there are a few peculiarities of evolution equations that are recalled here. It is conventional to denote the two independent variables in evolution equations $(x_1,x_2)=(t,x)$. Let
\begin{equation}\label{eq:GeneralEvolutionEquation}
 \mathcal E\colon u_t=G(t,x,u_0,\dots,u_r),\quad G_{u_r}\ne0,\quad r\geqslant2,
\end{equation}
be an evolution equation, where $u_k=\p_x^ku$, $k\in\mathbb N$, and $u_0:=u$. Without loss of generality, the conserved currents for evolution equations can be assumed to be independent of derivatives involving differentiations with respect to~$t$. In this specific situation, the definition~\eqref{eq:ConsLawDefinition} for a conserved current~$(\rho,\sigma)$ of the equation~$\mathcal E$ can be reduced to the identity
\begin{equation}\label{eq:ConsLawDefinitionEvolutionEquation}
\mathrm{{\bar D}}_t\rho+\mathrm D_x\sigma=0,\quad \mathrm{i.e.,}\quad \rho_t+\rho_*G+\mathrm D_x\sigma=0,
\end{equation}
where $\mathrm{{\bar D}}_t=\partial_t+(\mathrm D_x^kG)\partial_{u_k}$ is the restriction of the operator of total derivative with respect to~$t$ on the manifold defined by equation~$\mathcal E$ and its differential consequences in the jet space~$\mathrm J^\infty(t,x|u)$, and~$\rho_*=\rho_{u_k}\mathrm D^k_x$ is the Fr\'echet derivative of~$\rho$. The differential functions~$\rho$ and~$\sigma$ are the density and the flux of the conserved current $(\rho,\sigma)$. As evolution equations give the simplest example for systems in the extended Kovalevskaya form, the characteristic~$\lambda$ associated with the conservation law containing the conserved current~$(\rho,\sigma)$ can be expressed via the density~$\rho$~as
\[
 \lambda=\frac{\delta\rho}{\delta u},
\]
where $\delta/\delta u=(-\mathrm D_x)^k\p_{u_k}$ is the variational derivative, which coincides with the restriction of the Euler operator for differential functions that do not involve differentiation with respect to~$t$.

\begin{remark}\label{rem:OnConservedCurrentsForEvolutionEquations}
Prescribing an exact form of the conserved current~$(\rho,\sigma)$ defines the right hand side~$G$ of the corresponding evolution equation~$\mathcal E$ up to some number of arbitrary smooth functions of~$t$, and this number does not exceed~$\mathop{\rm ord}\rho$. Indeed, if $\rho$ and $\sigma$ are fixed differential functions, then the equation~\eqref{eq:ConsLawDefinitionEvolutionEquation} can be considered as $(\mathop{\rm ord}\rho)$th order inhomogeneous linear ordinary differential equation for~$G$ in total derivatives with respect to~$x$, where the variable~$t$ plays the role of a parameter. The general solution of this equation, if solutions exist at all, can be represented in the form $G=G^{\rm pi}+G^{\rm gh}$, where $G^{\rm pi}$ is a particular solution of this equation and $G^{\rm gh}$ is the general solution of the corresponding homogeneous equation $\rho_*G=0$. The solution set of the equation $\rho_*G=0$ coincides with the kernel of the operator~$\rho_*$. By Lemma~6.45 of~\cite{olve09Ay}, see also the proof of Theorem~1 in~\cite{soko88a}, the dimension of the kernel of~$\rho_*$ over the ring of smooth functions of~$t$ is not greater than~$\mathop{\rm ord}\rho$.
Moreover, only some pairs of differential functions can be conserved currents for evolution equations since the equation~\eqref{eq:ConsLawDefinitionEvolutionEquation} with prescribed values of~$\rho$ and $\sigma$ may have no solutions in~$G$.
The existence of evolution equations for a few prescribed candidates for conserved currents is even less expectable due to compatibility issues between the copies of the equation~\eqref{eq:ConsLawDefinitionEvolutionEquation} for the given pairs of $\rho$ and $\sigma$,
which also further reduces the possible arbitrariness in the form of~$G$.
\end{remark}

\begin{theorem}\label{thm:OnFormOfEvolutionEquations}
 An evolution equation of the form~\eqref{eq:GeneralEvolutionEquation} admits~$p$ linearly independent conservation laws with densities~$\rho^s$ and characteristics $\lambda^s=\delta \rho^s/\delta u$ if and only if
 \begin{equation}\label{eq:EvolutionEquationFormInverseCL}
  G=\mathrm{DT}(\lambda^1,\dots,\lambda^p)^\dag H- \sum_{s=1}^p\mathrm{DT}(\lambda^1,\dots,\cancel{\lambda^s},\dots,\lambda^p)^\dag \left(\frac{W(\lambda^1,\dots,\cancel{\lambda^s},\dots,\lambda^p)}{W(\lambda^1,\dots,\lambda^p)}\rho_t^s\right),
 \end{equation}
 where~$H$ is a differential function of $u$.
\end{theorem}

Theorem~\ref{thm:OnFormOfEvolutionEquations} is the natural extension of Theorem~\ref{thm:OnFormOfODE} to evolution equations. Note that in all Wronskians and Darboux transformations the derivatives are total derivatives with respect to~$x$.

Proposition~\ref{pro:RepresentationODESpanBasis} has a natural counterpart for evolution equations. Namely, the representation~\eqref{eq:EvolutionEquationFormInverseCL} depends on the linear span~$\langle\rho^1,\dots,\rho^p\rangle$ rather than on the fixed basis~$\{\rho^1,\dots,\rho^p\}$. Similarly, Proposition~\ref{pro:RepresentationODEUndeterminedness} can be extended to evolution equations, in that for fixed~$G$ and~$\langle\rho^1,\dots,\rho^p\rangle$, the general form of~$H$ is $H=H^0+f^s(t)\varphi^s$, where $H^0$ is a particular value of~$H$, $f^s(t)$ are arbitrary smooth functions of~$t$ and $\varphi^1,\dots,\varphi^p$ are adjoint functions to $\lambda^1,\dots,\lambda^p$ with respect to the independent variable~$x$.

Similar to~\eqref{eq:ODEFormInverseCLAlternative} we also have an alternative representation to~\eqref{eq:EvolutionEquationFormInverseCL}.
There are also order estimations for differential functions in~\eqref{eq:EvolutionEquationFormInverseCL} analogous to Lemma~\ref{lem:OnOrderOfRepresentation} and its corollaries.

\begin{example}
 To give an example for the representation~\eqref{eq:EvolutionEquationFormInverseCL}, consider the Korteweg--de Vries (KdV) equation
 \begin{equation}\label{eq:KdVEquationInverseProblem}
  u_t+uu_x+u_{xxx}=0.
 \end{equation}
It is well known that~\eqref{eq:KdVEquationInverseProblem} admits an infinite sequence of linearly independent conservation laws of growing order~\cite{miur68a}.
The two most elementary conservation laws for the KdV equation have the characteristics~$\lambda^1=1$, $\lambda^2=u$ and the conserved currents
\[
\left(u,\frac12u^2+u_{xx}\right),\quad
\left(\frac12u^2,\frac13u^3+uu_{xx}-\frac12u_x^2\right),
 \]
respectively.
For these two conserved currents and $G=-uu_x-u_{xxx}$,
the representation~\eqref{eq:EvolutionEquationFormInverseCL} becomes
\begin{equation}\label{eq:RepresentationKdVEquation}
G=\mathrm D_x\left(\mathrm D_x+\frac{u_{xx}}{u_x}\right)H,\quad H=-\frac12u_x-\frac{u^3}{6u_x}.
\end{equation}
Varying the function~$H$ in the expression~\eqref{eq:RepresentationKdVEquation} we obtain the representation for all evolution equations that admit the characteristics~$\lambda^1=1$ and $\lambda^2=u$.
\end{example}

\section[Inverse problems for infinite dimensional spaces of zeroth-order characteristics]
{Inverse problems for infinite dimensional spaces\\ of zeroth-order characteristics}\label{sec:ConservativeParameterizationWithDirectMethods}

We discuss the inverse problem on conservation laws for the case of a single partial differential equation of a single unknown function, i.e., $m=l=1$ and $n>1$,
that possesses a space of conservation-law characteristics parameterized by arbitrary smooth functions of one or several arguments
that themselves are functions of independent and dependent variables. 
Recall that such characteristics are typical for hydrodynamics models, see, e.g., Section~\ref{sec:ConservativeParameterizationVorticityEquation} below.  
Thus, in this section we consider a differential equation $\mathcal L\colon L[u]=0$ for the unknown function $u=u(x)$ of the variables $x=(x_1,\dots,x_n)$.

\begin{lemma}\label{lem:FirstLemmaCP}
A partial differential equation $\mathcal L\colon L[u]=0$ for a single unknown function~$u=u(x)$ of the variables $x=(x_1,\dots,x_n)$ admits the family of conservation-law characteristics $\{h(x_1)\}$, where $h$ runs through the set of smooth functions of~$x_1$,
if and only if the differential function~$L$ is represented in the form
\begin{equation}\label{eq:GeneralRepresentationForLemma1}
L[u]=\mathrm D_2F^2[u]+\cdots+\mathrm D_nF^n[u]
\end{equation}
for some differential functions $F^2$, \dots, $F^n$.
\end{lemma}

\begin{proof}
The ``if'' part is obvious.
Let us prove the ``only if'' part.
As the total divergence in the representation~\eqref{eq:GeneralRepresentationForLemma1} does not include the total derivative with respect to $x_1$, we can interpret the variable $x_1$ as a parameter and introduce $u^k:=\partial^ku/\partial x_1^k$, $k=0,\dots,N$ as dependent variables, where $N=\max\{\alpha_1\mid L_{u_\alpha}\ne0\}$.
In other words, $N$ is the highest number of differentiations with respect to $x_1$ that appears in derivatives involved in~$L[u]$.

Consider the expanded Euler operator $\hat{\mathsf E}=(\hat{\mathsf E}^0,\dots,\hat{\mathsf E}^N)$ with components
\[
\hat{\mathsf E}^k=(-\hat{\mathrm D})^{\hat \alpha}\frac{\p}{\p u^k_{(0,\hat\alpha)}}\,,\quad k=0,\dots,N.
\]
Here $(-\hat{\mathrm D})^{\hat \alpha}=(-\mathrm D_2)^{\alpha_2}\cdots (-\mathrm D_n)^{\alpha_n}$,
$\hat\alpha=(\alpha_2,\dots,\alpha_n)$ runs through the multi-index set~$\mathbb N_0^{\,\,n-1}$
and $u^k_{(0,\hat\alpha)}$ is the jet variable identified with $\partial^{k+|\hat\alpha|}u/\partial x_1^k\partial x_2^{\alpha_2}\cdots\partial x_n^{\alpha_n}$.
Thus, here and in what follows the hat over a tuple symbol denotes the exclusion of the first component.
It then suffices to prove that the action of the operator $\hat{\mathsf E}$ on $\mathcal L$ is zero,
 \[
 \hat{\mathsf E}^kL=(-\hat{\mathrm D})^{\hat \alpha}L_{u^k_{(0,\hat\alpha)}}=0,\quad k=0,\dots,N.
 \]

Let $\mathsf E$ be the standard Euler operator with respect to~$u$, $\mathsf E=(-\mathrm D)^\alpha\p_{u_\alpha}$.
If the function $h=h(x_1)$ is a conservation-law characteristic for the equation $L[u]=0$, then we have $\mathsf E(hL)=0$, which can be expanded in the following way:
\begin{align}\nonumber
 \mathsf E(hL)&=(-\mathrm D)^\alpha(hL)_{u_\alpha}=
 (-\mathrm D_1)^{\alpha_1}(-\hat{\mathrm D})^{\hat\alpha}(hL)_{u^{\alpha_1}_{\hat\alpha}}=(-\mathrm D_1)^{\alpha_1}(h\hat{\mathsf E}^{\alpha_1}L)
\\  \nonumber
 &=\sum_{\alpha_1=0}^N(-1)^{\alpha_1}\sum_{s=0}^{\alpha_1}\binom{\alpha_1}s (\partial^s_{x_1}h)\mathrm D_1^{\alpha_1-s}(\hat{\mathsf E}^{\alpha_1}L)
\\ \label{eq:GeneralRepresentationForLemma1SystemForSplitting}
 &=\sum_{s=0}^N(\partial_{x_1}^sh)\sum_{\alpha_1=s}^N(-1)^{\alpha_1}\binom{\alpha_1}s \mathrm D_1^{\alpha_1-s}(\hat{\mathsf E}^{\alpha_1}L)=0.
\end{align}
Here we used the definition of the expanded Euler operator~$\hat{\mathsf E}$ and the fact that $h$ depends on~$x_1$ only. Since the function $h$ is arbitrary, the last equality in~\eqref{eq:GeneralRepresentationForLemma1SystemForSplitting} can be split with respect to the various derivatives of $h$, which gives the system
\begin{equation}\label{eq:SplitEqForLwithChar_h(x_1)}
 \sum_{\alpha_1=s}^N(-1)^{\alpha_1}\binom{\alpha_1}s \mathrm D_1^{\alpha_1-s}(\hat{\mathsf E}^{\alpha_1}L)=0,\quad s=0,\dots,N.
\end{equation}
We start with the highest value~$s=N$ and proceed to the lower values of~$s$
using the results for the higher values.
Thus, the equation~\eqref{eq:SplitEqForLwithChar_h(x_1)} with~$s=N$ is $\hat{\mathsf E}^NL=0$.
Using this result in the equation~\eqref{eq:SplitEqForLwithChar_h(x_1)} with $s=N-1$
and continuing in a similar way up to $s=0$, we obtain the simplified system
\[
\hat{\mathsf E}^NL=0,\quad \hat{\mathsf E}^{N-1}L=0,\quad\dots,\quad\hat{\mathsf E}^0L=0.
\]
In view of~\cite[Theorem~4.7]{olve93Ay},
this system implies that $L[u]$ is of the form~\eqref{eq:GeneralRepresentationForLemma1},
which completes the proof.
\end{proof}

\begin{corollary}
 If a differential equation $\mathcal L\colon L[u]=0$ admits $N+1$ locally linearly independent conservation laws with characteristics depending only on the single variable $x_1$, $h^{s'}=h^{s'}(x_1)$, $s'=0,\dots,N$, then the representation~\eqref{eq:GeneralRepresentationForLemma1} holds and hence the equation admits conservation laws with characteristics being arbitrary smooth functions of $x_1$.
\end{corollary}

\begin{proof}
Substituting the characteristics $h^{s'}$ into~\eqref{eq:GeneralRepresentationForLemma1SystemForSplitting}, we obtain the condition
\[
 \sum_{s=0}^N(\partial_{x_1}^sh^{s'})\sum_{\alpha_1=s}^N(-1)^{\alpha_1}\binom{\alpha_1}s \mathrm D_1^{\alpha_1-s}(\hat{\mathsf E}^{\alpha_1}L)=0,\quad s'=0,\dots,N.
\]
Since the functions $h^{s'}$ are linearly independent and hence their Wronskian is (locally) nonvanishing, $|\partial^s_{x_1}h^{s'}|_{s,s'=0,\dots,N}\ne0$,
the above condition implies the system~\eqref{eq:SplitEqForLwithChar_h(x_1)}.
Thus, the further proof is the same as the proof of Lemma~\ref{lem:FirstLemmaCP}.
\end{proof}

A further direct corollary is the following:

\begin{corollary}\label{cor:DEsAdmittingArbFunctionsOfSeveralVariablesAsChars}
A differential equation $\mathcal L\colon L[u]=0$ admits an arbitrary smooth function $h=h(x_1,\dots,x_l)$ as a conservation-law characteristic, where $l<n$,
if and only if the differential function~$L$ is represented in the form
\[
  L=\mathrm D_{x_{l+1}}F^{l+1}[u]+\dots+\mathrm D_nF^n[u],
\]
for a tuple of $n-l$ differential functions $F^{l+1},\dots,F^n$.
\end{corollary}

\begin{remark}
The case $l=n$ is singular for the formulation of Corollary~\ref{cor:DEsAdmittingArbFunctionsOfSeveralVariablesAsChars}.
The corresponding assertion is called the du Bois-Reymond lemma \cite[Lemma~5.67]{olve93Ay}. 
It states the following:
If for any function $h=h(x)$ there exists an $n$-tuple of differential functions~$F[u]$ such that $hL[u]=\mathop{\rm Div}\nolimits F[u]$,
then the differential function~$L[u]$ is in fact a function of~$x$ alone.
\end{remark}

\begin{lemma}\label{lem:SecondLemmaCP}
A differential equation $\mathcal L\colon L[u]=0$ for the unknown function $u=u(x)$, $x=(x_1,\dots,x_n)$
admits the characteristics $h(x_1)$ and $f(x_1)x_2$, where $h$ and~$f$ run through the set of smooth functions of~$x_1$, if and only if its right hand side~$L$ is represented in the form
\begin{equation}\label{eq:GeneralRepresentationForLemma2}
L=\mathrm D_2^2G^2[u]+\mathrm D_3G^3[u]+ \cdots+\mathrm D_nG^n[u],
\end{equation}
for some differential functions $G^2[u]$, \dots, $G^n[u]$.
\end{lemma}

\begin{proof}
We again prove the ``only if'' part merely.
Lemma~\ref{lem:FirstLemmaCP} implies the representation $L=\mathop{\widehat{\rm Div}}\hat F$
for some tuple $\hat F=(F^2[u],\dots,F^n[u])$,
where $\mathop{\widehat{\rm Div}}\hat F=\mathrm D_2F^2+\cdots+\mathrm D_nF^n$.
As $f(x_1)x_2$ is a conservation-law characteristic of~$\mathcal L$, we also have
 \begin{equation}\label{eq:ProofOfLemma2}
  fx_2L=fx_2\mathop{\widehat{\rm Div}}\hat F=\mathop{\widehat{\rm Div}}(fx_2\hat F)-fF^2=\mathop{\rm Div}H,
 \end{equation}
where we used `integration by parts' to arrive at the second equality and the definition of conser\-vation-law characteristic to establish the third equality. Here the tuple $H=(H^1[u],\dots,H^n[u])$ is a conserved current associated with the characteristic $f(x_1)x_2$. From the last equality in~\eqref{eq:ProofOfLemma2} we conclude that
 \[
  fF^2=\mathop{\widehat{\rm Div}}(fx_2\hat F)-\mathop{\rm Div}H,
 \]
that is, the arbitrary smooth function $f=f(x_1)$ is a conservation-law characteristic for the equation $F^2[u]=0$.
In view of Lemma~\ref{lem:FirstLemmaCP}, the differential function~$F^2$ admits the representation $F^2=\mathop{\widehat{\rm Div}}\hat K[u]$ for some tuple of differential functions $\hat K=(K^2[u],\dots,K^n[u])$.
Therefore,
\[
L=\mathrm D_2\mathop{\widehat{\rm Div}}\hat K[u]+\mathrm D_3F^3[u]+\cdots+\mathrm D_nF^n[u],
\]
which leads to the representation~\eqref{eq:GeneralRepresentationForLemma2} with $G^2=K^2$ and $G^j=F^j+\mathrm D_2K^j$, $j=3,\dots,n$.
\end{proof}

\begin{corollary}\label{cor:FormOfEqsWithHugeNumberOfCLs}
A differential equation $\mathcal L\colon L[u]=0$ for the unknown function $u=u(x)$ of the variables $x=(x_1,\dots,x_n)$
admits characteristics of the form $h(x_1)+\sum_{i=2}^nf^i(x_1)x_i$ with arbitrary smooth functions $h=h(x_1)$ and $f^i=f^i(x_1)$, $i=2,\dots,n$,
if and only if the differential function~$L$ can be represented as
\begin{equation}\label{eq:RepresentationForLWithChars_hfixi}
L=\sum_{i,j=2}^n\mathrm D_i\mathrm D_jK^{ij}[u].
\end{equation}

\end{corollary}

\begin{proof}
Following the proof of Lemma~\ref{lem:SecondLemmaCP},
we obtain that $L=\mathop{\widehat{\rm Div}}\hat F$
with $F^i=\mathop{\widehat{\rm Div}}\hat K^i$ for some $\hat K^i=(K^{i2}[u],\dots,K^{in}[u])$, $i=2,\dots,n$.
\end{proof}

We can rearrange the representation~\eqref{cor:FormOfEqsWithHugeNumberOfCLs}
and assume that $K^{ij}$'s are symmetric in $(i,j)$, $K^{ij}=K^{ji}$,
or that the summation range for $(i,j)$ is $2\leqslant i\leqslant j\leqslant n$.

Lemma~\ref{lem:FirstLemmaCP} can also be generalized by considering arbitrary functions of more general arguments.

\begin{lemma}\label{lem:PDEsWithCharsBeingArbitraryFunctionOfOmega}
A partial differential equation $\mathcal L\colon L[u]=0$ for a single unknown function~$u=u(x)$ of the variables $x=(x_1,\dots,x_n)$
admits the family of conservation-law characteristics $\{h(\omega)\}$,
where $h$ runs through the set of smooth functions of~$\omega=\omega(x,u)$ being a nonconstant smooth function of~$(x,u)$,
if and only if the differential function~$L$ is represented in the form
\begin{equation}\label{eq:GeneralRepresentationForPDEsWithCharsBeingArbitraryFunctionOfOmega}
L[u]=\mathrm D_i(G^{ij}\mathrm D_j\omega)=(\mathrm D_iG^{ij})\mathrm D_j\omega=\mathrm D_j(\omega\mathrm D_iG^{ij})
\end{equation}
for some differential functions $G^{ij}=G^{ij}[u]$ with $G^{ij}=-G^{ji}$.
\end{lemma}

\begin{proof}
Since the ``if'' part of the lemma is obvious, we merely prove the ``only if'' part.
The above condition on $\mathcal L$ means that $hL\in\mathop{\rm im}\nolimits\mathop{\rm Div}\nolimits$
for any smooth function~$h=h(\omega)$.
Without loss of generality, we can assume that $\mathrm D_1\omega\ne0$.
Locally changing coordinates
$\tilde x_1=\omega$, $\tilde x_i=x_i$, $i\ne1$, 
and either $\tilde u=u$ if $\omega_u=0$ or $\tilde u=x_1$ if $\omega_u\ne0$.
In the new coordinates, lemma's supposition takes the form
$JhL\in\mathop{\rm im}\nolimits\mathop{\widetilde{\rm Div}}\nolimits$,
where $J=\det(\mathrm D_j\tilde x_i)_{i,j=1,\dots,n}$,
and $\mathop{\widetilde{\rm Div}}\nolimits$ is the divergence in the total derivatives~$\tilde{\mathrm D}_j$ 
with respect to the new coordinates $\tilde x_j$;
cf.\ \cite[Proposition~1]{popo2008a}.
Then Lemma~\ref{lem:FirstLemmaCP} implies that $JL=\mathop{\widetilde{\rm Div}}\nolimits\tilde F$
with the tuple $\tilde F=(\tilde F^1[\tilde u],\dots,\tilde F^n[\tilde u])$, where $\tilde F^1=0$.
Returning to the old coordinates, we obtain
$L=\mathop{\rm Div}\nolimits F$, where the tuples $F=(F^1[u],\dots,F^n[u])$ and~$\tilde F$
are related by $J\tilde F^i=F^j\mathrm D_j\tilde x_i$.
Therefore, the equality $\tilde F^1=0$ is equivalent to the equation $F^j\mathrm D_j\omega=0$.
The symmetric representation of the general solution of this equation with respect to~$F$
is $F^i=G^{ij}\mathrm D_j\omega$
for some differential functions $G^{ij}=G^{ij}[u]$ with $G^{ij}=-G^{ji}$.
\end{proof}

The symmetrization of the representation~\eqref{eq:GeneralRepresentationForPDEsWithCharsBeingArbitraryFunctionOfOmega}
leads to an additional ambiguity of the coefficients~$G^{ij}$,
and thus they can be constrained more.
In particular, if $\mathrm D_1\omega\ne0$, then we can set $G^{ij}=0$, $i,j\ne1$.

The representation~\eqref{eq:GeneralRepresentationForPDEsWithCharsBeingArbitraryFunctionOfOmega}
can also be rewritten in the form
$L[u]=F^j\mathrm D_j\omega=\mathrm D_j(\omega F^j)$,
where $F=(F^1[u],\dots,F^n[u])$ is a null divergence, $\mathop{\rm Div}F=0$.

\begin{remark}\label{rem:PDEsWithCharsBeingArbitraryDiffFunctionOfOmega}
The ``if'' part of Lemma~\ref{lem:PDEsWithCharsBeingArbitraryFunctionOfOmega}
is obviously true even for~$\omega$ being a nonconstant differential function of~$u$.
We may conjecture that the ``only if'' part is also true
although its proof needs tools more powerful than and different from those
used for proving Lemma~\ref{lem:PDEsWithCharsBeingArbitraryFunctionOfOmega}.
\end{remark}

\section[Conservative parameterization as inverse problem on conservation laws]
{Conservative parameterization as inverse problem\\ on conservation laws}\label{sec:ConservativeParameterizationSchemes}

Despite the ever increase in computational power, it is impossible to run a numerical simulation at infinite resolution.
In other words, dictated by the computational resources available and computational costs acceptable, one has to choose a \emph{maximum} resolution,
which in turn introduces a \emph{minimal} scale below that the model is not capable of resolving physical processes any more.
At the same time, nonlinear systems are characterized by the interaction of various scales
and, therefore, the effects of the unresolved \emph{subgrid-scale} processes  on the resolved processes cannot be omitted in a numerical simulation of such systems.
Modeling these effects is known as \emph{parameterization}.

\looseness=1
In developing a numerical approximation including the construction of parameterization schemes,
it is in general important to maintain the consistency with the original model.
A~major challenge in the construction of parameterization schemes for differential equations
is to ensure the preservation of important geometric structures of the original system~$\mathcal L$ of governing differential equations.
By this it is meant that a closure model for unresolved processes
should lead to a system of differential equations for the resolved quantities
that shares some of the features of the system~$\mathcal L$.
Both symmetries and conservation laws play a fundamental role for the initial formulation of physical theories and their study.
This is why they may be among such shared features that should be preserved even on the level of the resolved scales of the system~$\mathcal L$.

Research of symmetry-preserving closure models was initiated in~\cite{ober97Ay} for the Navier--Stokes equations and formalized in~\cite{popo10Cy} using the language of group analysis of differential equation. In the latter paper it was demonstrated that finding \emph{local parameterization schemes} (i.e., parameterization schemes that model the unresolved processes at a point by using only the information of the resolved processes in this point) can be re-cast as a group classification problem. This observation unlocks the use of a variety of techniques from the group classification of differential equations to be applied to the parameterization problem. See~\cite{basa01Ay,bihl11Dy,ovsi82Ay,popo10Ay} and references therein for a discussion of various group classification methods.

The ideology of conservative parameterization is similar to that of symmetry-preserving closure models.
If the system $\mathcal L$ describes, e.g., an energy preserving physical process that cannot be resolved numerically,
then a parameterization for this process should still preserve energy.
This is essentially the problem of finding conservative parameterization schemes.
First examples for conservative parameterization schemes were given in~\cite{bihl12Dy,bihl11Fy}.
Here we interpret the inverse problem on conservation laws in the light of the conservative parameterization
using, for the sake of simplicity, local first-order parameterization schemes.

Mathematically, the splitting into resolved and unresolved scales is done by decomposing the unknown functions as a mean part and a deviation part, i.e., $u=\bar u+u'$, where bars are used to denote means and primes denote the deviations from this mean. In order to derive the model for the mean part~$\bar u$, one has to average $\mathcal L$. Depending on the averaging rule invoked and the particular form of the system $\mathcal L$, the resulting averaged system~$\bar{\mathcal L}$ is typically not closed.
In other words, additionally to derivatives of~$\bar u$ the system~$\bar{\mathcal L}$ involves a tuple of subgrid-scale quantities $w=(w^1,\dots,w^k)$,
\[
\bar{\mathcal L}\colon\quad \bar L^\mu(x,\bar u_{(\bar r)},w)=0,\quad \mu=1,\dots,l,
\]
and thus it can only be used in practice once a parameterization for~$w$ is found.

\begin{definition}
A \emph{local first-order parameterization scheme} is an expression of the tuple of unknown subgrid-scale quantities $w$ as differential functions of~$\bar u$,
$
w=f[\bar u]$,
where $f=(f^1,\dots,f^k)$ are the \emph{parameterization functions}~\cite{stul88Ay}.
\end{definition}

Recall that in a higher order parameterization scheme one typically has to extend the averaged system with differential equations for the unresolved quantities. This in turn introduces new unresolved quantities of higher order, which should also be parameterized. The extension of the initial system leads to several complications in the course of study of structure-preserving parameterization schemes, which will be the subject of future investigations.

As a start for the consideration, the parameterization functions can be assumed as arbitrary differential functions of certain order, or one can use a specific ansatz that depends on further arbitrary differential functions~$\tilde f$ satisfying certain differential constraints.
Upon inserting the local first-order parameterization scheme $w=f[\bar u]$ into the unclosed model $\bar{\mathcal L}$ one obtains a class of closed systems of differential equations, denoted by~$\bar{\mathcal L}|_{w=f}$, in which $\tilde f$ acts as tuple of arbitrary elements of the class. This closed system of differential equations can now be used in practice as it is a system for $\bar u$ only. The main task is to find the arbitrary elements~$\tilde f$ such that~$\bar{\mathcal L}|_{w=f}$ adequately describes the evolution of $\bar u$. This is the parameterization problem in general. Here we are interested in solving this problem using conservation laws of differential equations.

\begin{definition}
 A local parameterization scheme is called \emph{conservative} if the parameterized system of differential equations admits nontrivial conservation laws.
\end{definition}

As formulated above, the problem of finding conservative parameterization schemes can be tackled within the context of both the direct and the inverse problems on conservation laws. Within the framework of the direct problem, we would start with a general closed class of differential equations with yet to be determined parameterization functions, and aim to classify, up to an equivalence, the forms of the parameterization functions leading to equations from the closed class that admit nontrivial conservation laws. Within the framework of the inverse problem, which is the subject of the present paper, one starts with conservation laws admitted by the original system~$\mathcal L$ and constructs the parameterization functions $f$ in such a way that the closed system~$\bar{\mathcal L}|_{w=f}$ admits the corresponding conservation laws. For more information, see~\cite{bihl12Dy,bihl11Fy}.

\section{Conservative parameterization for the vorticity equation}\label{sec:ConservativeParameterizationVorticityEquation}

We demonstrate the procedure of conservative closure introduced in Section~\ref{sec:ConservativeParameterizationWithDirectMethods}
for the system of two-dimensional incompressible Euler equations.
This example is particularly challenging as the conservation laws to be preserved correspond to characteristics parameterized by \emph{arbitrary} functions of certain arguments.
The two-dimensional incompressible Euler equations in the stream function form
reduce to the single vorticity equation
\begin{equation}\label{eq:VorticityEquationCP}
 \zeta_t + \psi_x\zeta_y-\psi_y\zeta_x=0,\qquad \zeta:=\psi_{xx}+\psi_{yy}.
\end{equation}
Here $\psi=\psi(t,x,y)$ is the stream function generating horizontal, two-dimensional incompressible flow $\mathbf v=\mathbf k\times\nabla\psi$ with $\mathbf k=(0,0,1)^{\sf T}$,  $\zeta$ is the vorticity, and we use a specific notation for the variables,  $t$, $x$, $y$ and $\psi$ instead of $x_1$, $x_2$, $x_3$ and~$u$.

Subsequently we are interested in the dynamics of the mean part~\eqref{eq:VorticityEquationCP} as this is the only part that is accessible to a numerical model, which usually operates on a fixed, finite resolution.
That is, we split the dependent variable $\psi$ as well as $\mathbf v$ and~$\zeta$ into resolved (mean) parts $\bar\psi$, $\bar{\mathbf v}$ and $\bar\zeta$ and unresolved (sub-grid scale) parts $\psi'$, $\mathbf v'$ and $\zeta'$, i.e.,
\[\psi=\bar\psi+\psi',\quad \mathbf v=\bar{\mathbf v}+\mathbf v' \quad\mbox{and}\quad \zeta=\bar\zeta +\zeta'.\]
Averaging Eq.~\eqref{eq:VorticityEquationCP} in a way satisfying the Reynolds averaging rule $\overline{ab}=\bar a\bar b+\overline{a'b'}$ (e.g.\ using the ensemble average),
we obtain the equation for the mean part, which is
\begin{equation}\label{eq:VorticityEquationAveragedCP}
 \bar\zeta_t + \bar\psi_x\bar\zeta_y-\bar\psi_y\bar\zeta_x=\nabla\cdot\overline{\zeta'\mathbf v'},\qquad \bar\zeta:=\bar\psi_{xx}+\bar\psi_{yy}.
\end{equation}
The problem with the averaged equation~\eqref{eq:VorticityEquationAveragedCP} is that its right-hand side involves the divergence of the unknown \emph{vorticity flux} $\overline{\zeta'\mathbf v'}$ for which no equation is given. In other words, system~\eqref{eq:VorticityEquationAveragedCP} is underdetermined, which is the usual closure problem of fluid mechanics.

We close the equation~\eqref{eq:VorticityEquationAveragedCP} by assuming
a functional relation between the vorticity flux, the independent variables, the averaged stream function~$\bar\psi$ and its various derivatives,
\[
 \nabla\cdot\overline{\zeta'\mathbf v'}=V[\bar\psi].
\]
Here the notation $V=V[\bar\psi]$ indicates that $V$ is a differential function of $\bar\psi$,
i.e., a smooth function of $t$, $x$, $y$, $\bar\psi$ and derivatives of~$\bar\psi$ with respect to $t$, $x$ and~$y$.
Introducing this closure ansatz in~\eqref{eq:VorticityEquationAveragedCP} leads to
\begin{equation}\label{eq:VorticityEquationAveragedClosedCP}
 \zeta_t + \psi_x\zeta_y-\psi_y\zeta_x=V[\psi],\qquad \zeta:=\psi_{xx}+\psi_{yy},
\end{equation}
where we have omitted bars over the averaged quantity for the sake of notational simplicity. As this system is now closed, i.e., it does not involve unresolved quantities any more, this notation is consistent. The problem is now to specify the functional form of $V$ in such a manner that the equation~\eqref{eq:VorticityEquationAveragedClosedCP} preserves some of the conservation laws of the original vorticity equation~\eqref{eq:VorticityEquationCP}.

More specifically, we are interested in preserving the conservation laws of~\eqref{eq:VorticityEquationCP} with zeroth-order characteristics, which are in the span
\begin{equation}\label{eq:0th-orderCharSpaceOfVorticityEq}
\langle h(t),\, f(t)x,\, g(t)y,\, \psi\rangle.
\end{equation}
Physically, these conservation-law characteristics are associated with the conservation of generalized circulation, generalized momenta in $x$- and $y$-directions and energy, respectively. All of these conservation laws are of superior importance in fluid mechanics and hence it is natural to attempt finding parameterization schemes of the general form~\eqref{eq:VorticityEquationAveragedClosedCP} that preserve these conservation laws in the averaged model.

In order to determine the functional form of $V$,
we can use the theory laid down in Section~\ref{sec:ConservativeParameterizationWithDirectMethods}.
Before constructing $V$, it is instructive to check directly
that the vorticity equation~\eqref{eq:VorticityEquationCP} itself can be brought
into the form required by Corollary~\ref{cor:FormOfEqsWithHugeNumberOfCLs}.
Thus, the repeated `integration by parts' of the Jacobian $\psi_x\zeta_y-\psi_y\zeta_x$ leads to the expression
\[
 \psi_x\zeta_y-\psi_y\zeta_x=(\mathrm D_y^2-\mathrm D_x^2)(\psi_x\psi_y)+ \mathrm D_x\mathrm D_y(\psi_x^2-\psi_y^2).
\]
Then, we can represent the left hand side of the vorticity equation~\eqref{eq:VorticityEquationCP} as
\[
 \zeta_t+\psi_x\zeta_y-\psi_y\zeta_x=\mathrm D_x^2(\psi_t-\psi_x\psi_y)+ \mathrm D_x\mathrm D_y(\psi_x^2-\psi_y^2)+\mathrm D_y^2(\psi_t+\psi_x\psi_y),
\]
which is in accordance with Corollary~\ref{cor:FormOfEqsWithHugeNumberOfCLs}.
In view of this corollary, a similar representation also holds for the differential function $V[\psi]$, i.e.,
\begin{equation}\label{eq:GenConservativeParameterizationForV}
V[\psi]=\mathrm D_x^2F^{11}+\mathrm D_x\mathrm D_yF^{12}+\mathrm D_y^2F^{22}.
\end{equation}
with some differential functions $F^{11}$, $F^{12}$ and $F^{22}$ of $\psi$,
if and only if the closed vorticity equation~\eqref{eq:VorticityEquationAveragedClosedCP}
still admits generalized circulation and momenta as conservation laws.

The remaining task is now to find the restricted form of $V$ if a parameterization for the vorticity equation should also preserve energy, i.e., the characteristic $\psi$. In this case
\[
 \psi V=\psi(\mathrm D_x^2F^{11}+\mathrm D_x\mathrm D_yF^{12}+ \mathrm D_y^2F^{22})=\mathop{\rm Div}H
\]
for some triple of differential functions~$H$.
Using `integration by parts' in the above equation, we obtain
\[
 \psi_{xx}F^{11}+\psi_{xy}F^{12}+\psi_{yy}F^{22}=\mathop{\rm Div}Q,
\]
for another triple of differential functions $Q$. This is a single inhomogeneous linear algebraic equation for the components $F^{11}$, $F^{12}$ and $F^{22}$. The solution of this equation can be represented in a symmetric way as
\begin{align*}
 &F^{11}=\psi_{yy}P^2-\psi_{xy}P^3+R^1,\\
 &F^{12}=\psi_{xx}P^3-\psi_{yy}P^1+R^2,\\
 &F^{22}=\psi_{xy}P^1-\psi_{xx}P^2+R^3,
\end{align*}
where $P^i=P^i[\psi]$ are arbitrary differential functions of $\psi$, $i=1,2,3$, and
the triple of differential functions $R^i=R^i[\psi]$ is a particular solution of the equation,
\[
 \psi_{xx}R^1+\psi_{xy}R^2+\psi_{yy}R^3=\mathop{\rm Div}Q.
\]
A nice particular solution satisfies the additional constraints $R^2=0$, $R^1=R^3$, which gives
\[R^2=0,\quad R^1=R^3=\frac1\zeta\mathop{\rm Div}Q.\]
The possible singularity in points where the vorticity vanishes can be compensated by vanishing $\mathop{\rm Div}Q$ in the same points.
For example, if $Q=\zeta^2S$ for some triple $S$ of differential functions of~$\psi$, then
$\mathop{\rm Div}Q=\zeta^2\mathop{\rm Div}S+2\zeta(S^1\zeta_t+S^2\zeta_x+S^3\zeta_y)$.

The substitution of the above solution into~\eqref{eq:GenConservativeParameterizationForV} leads to the following assertion:

\begin{proposition}\label{pro:ConservativeParameterizationOfVortisityEq}
If the unclosed vorticity flux $\nabla\cdot\overline{\zeta'\mathbf v'}$ is parameterized by $V[\psi]$ in the Reynolds averaged vorticity equation, where $V[\psi]$ is given through
\begin{gather*}
V[\psi]=\mathrm D_x^2(\psi_{yy}P^2-\psi_{xy}P^3)+ \mathrm D_x\mathrm D_y(\psi_{xx}P^3-\psi_{yy}P^1)+\mathrm D_y^2(\psi_{xy}P^1- \psi_{xx}P^2)\\
\phantom{V[\psi]=}{}+(\mathrm D_x^2+\mathrm D_y^2)(\zeta\mathop{\rm Div}S+2S^1\zeta_t+2S^2\zeta_x+2S^3\zeta_y)
\end{gather*}
for some differential functions $P^i$ and $S^i$, $i=1,2,3$, of $\psi$ the resulting closed equation~\eqref{eq:VorticityEquationAveragedClosedCP} possesses the conservation laws associated with characteristics $h(t)$, $f(t)x$, $h(t)y$ and $\psi$. That is, the closed equation will preserve generalized circulation, generalized momenta in $x$- and $y$-direction and energy.
\end{proposition}

The expression for~$V$ given in Proposition~\ref{pro:ConservativeParameterizationOfVortisityEq}
is still too general.
It can be considered as ansatz for~$V$ that should be further specified from the physical point of view.
Various additional constraints can be imposed on~$P^i$ and $S^i$
for the parameterized equation~\eqref{eq:VorticityEquationAveragedClosedCP}
to possess other required properties such as the preservation of Lie symmetries or the consistency with the structure of the initial equation~\eqref{eq:VorticityEquationCP}.
In particular, we can set $S^1=0$ and choose the other differential functions~$P^i$ and~$S^i$
not to depend on derivatives of~$\psi$ that involve differentiation with respect to~$t$.
Then $V$ has the same property.

\section{Research perspective for inverse problem on conservation laws}\label{sec:ConclusionCP}

In this paper we have, for the first time, explicitly introduced the inverse problem on conservation laws and rigorously stated it in important particular cases. As for the analogous inverse problem on group classification, posing this problem is motivated by several practical applications.

We have solved the inverse problem on conservation laws for single ordinary differential equations, for single evolution equations and for single general partial differential equations that admit infinite dimensional spaces of zeroth-order characteristics of simple structure.
Nevertheless, the methods introduced here can be carried over to the more complicated case of systems of differential equations. A main difficulty while extending the results derived in this paper to the case of systems of differential equations is that the formulas involved will become more cumbersome, but this poses mostly a technical hurdle.

There are several additional areas of importance in this subject that warrant further investigation.

A problem of high practical relevance is to study partial differential equations with general infinite dimensional spaces of conservation-law characteristics. From the viewpoint of real-world applications it is required to study conservative parameterization schemes that preserve conservation laws with characteristics of order higher than zero. For example, it is well know that in addition to the zeroth-order characteristics~\eqref{eq:0th-orderCharSpaceOfVorticityEq} the vorticity equation~\eqref{eq:VorticityEquationCP} also admits the family of characteristics $\{h(\zeta)\}$, where $h$ is an arbitrary smooth function of the vorticity~$\zeta=\psi_{xx}+\psi_{yy}$. At least some of these characteristics, e.g.~$h=\zeta$ corresponding to preservation of enstrophy, are physically essential.
The left hand side of the vorticity equation~\eqref{eq:VorticityEquationCP}
can be represented in the form~\eqref{eq:GeneralRepresentationForPDEsWithCharsBeingArbitraryFunctionOfOmega},
where $(x_1,x_2,x_3)=(t,x,y)$, $u=\psi$, $\omega=\zeta$ and
\[
(G^{ij}[\psi])=\begin{pmatrix}0&0&y\\0&0&-\psi\\-y&\psi&0\end{pmatrix},
\]
cf.\ Remark~\ref{rem:PDEsWithCharsBeingArbitraryDiffFunctionOfOmega}.
A problem is to prove that any equation possessing the family of characteristics $\{h(\zeta)\}$
can be represented in the same form with certain functions $G^{ij}=-G^{ji}$.
Then, to describe equations admitting all the above characteristics of the vorticity equation,
one should merge this representation
with the representation given in Proposition~\ref{pro:ConservativeParameterizationOfVortisityEq},
which is not trivial.

The parameterization schemes constructed in Section~\ref{sec:ConservativeParameterizationVorticityEquation} for the eddy-vorticity flux in the incompressible Euler equations give, to the best of our knowledge, the first example for systematically finding closure schemes that lead to parameterized models preserving conservation laws of the initial model. At the same time, the obtained parameterizations are physically quite simple as they are of first order. Modern state-of-the-art parameterization schemes depend not only on the mean values of the equation variables but can include higher-order correlation terms as well. An example for a symmetry-preserving higher-order parameterization schemes was recently given in~\cite{bihl15a}. Extending the methods for finding conservative parameterization schemes to more complicated closure ansatzes will be an important problem of future investigations.

Although it is natural to consider the inverse problem on conservation laws as a problem in itself, in practical applications it might be relevant to construct systems of differential equations that possess both certain conservation laws and a prescribed symmetry group. In other words, it will be necessary to consider the \emph{joint inverse problem on conservation laws and group classification}, cf.\ \cite{ande1994b,ande1995c}. This is important in the study of parametrization schemes as not only conservation laws have to be respected when deriving a closure model. Other geometric properties as for example Lie symmetries might have to be preserved as well, as discussed e.g.\ in~\cite{bihl12Dy,bihl11Fy,ober97Ay,popo10Cy,sten07Ay,stul88Ay}.

Lastly, another relevant study will be the inverse problem on conservation laws for difference equations. This problem is intimately linked to the problem on conservative discretization, which is of immediate practical relevance. First results on applying the mathematical machinery on conservation laws to finding conservative discretization schemes were given in~\cite{wan13Ay}, and more extended investigations on this problem are currently underway.

\section*{Acknowledgements}

The authors are grateful to the anonymous referees and the editor for a number of valuable remarks and suggestions. 
The authors also thank Michael Kunzinger, Artur Sergyeyev and Elsa Maria Dos Santos Cardoso-Bihlo for helpful discussions.
This research was undertaken, in part, thanks to funding from the Canada Research Chairs program, the NSERC Discovery Grant program and the InnovateNL LeverageR{\&}D program. 
The research of ROP was supported by the Austrian Science Fund (FWF), project P25064. 

\footnotesize\setlength{\itemsep}{0ex}

\end{document}